%% file: main.tex
\documentclass[letterpaper,11pt]{article}
\pdfoutput=1 
\usepackage[utf8]{inputenc}

\usepackage[margin=1in]{geometry}

\input{macros}

\title{Contest Design with Threshold Objectives}

\author[2]{Edith Elkind}
\author[1]{Abheek Ghosh}
\author[1]{Paul W.\ Goldberg}
\affil[1]{University of Oxford}
\affil[2]{Northwestern University}
\date{}

\begin{document}

\maketitle

\begin{abstract}
We study contests where the designer's objective is an extension of the widely studied objective of maximizing the total output: The designer gets zero marginal utility from a player's output if the output of the player is very low or very high. We consider two variants of this setting, which correspond to two objective functions: \textit{binary threshold}, where the designer's utility is a non-decreasing function of the number of players with output above a certain threshold; and \textit{linear threshold}, where a player's contribution to the designer's utility is linear in her output if the output is between a lower and an upper threshold, and becomes constant below the lower and above the upper threshold. For both of these objectives, we study \textit{rank-order allocation} contests and \textit{general} contests.
We characterize the contests that maximize the designer's objective and indicate techniques to efficiently compute them.
\end{abstract}

\begingroup
\renewcommand\thefootnote{}\footnote{This paper was published in the \textit{International Journal of Game Theory}, 2025. A preliminary version was presented at the 17th Conference on Web and Internet Economics, 2021.}%
\addtocounter{footnote}{-1}
\endgroup

\section{Introduction}
Contests are games in which (1) players invest effort and produce outputs toward winning one or more prizes, (2) those investments of effort are costly and irreversible, and (3) the prizes are allocated based on the values of outputs. They are prevalent in many areas, including sports, rent-seeking, patent races, innovation inducement, labor markets, college admissions, scientific projects, crowdsourcing and other online services.\footnote{See the book by Vojnovic~\cite{vojnovic2015contest} for an introduction to contest theory and relevant resources.} 

We study contests as incomplete information games. We assume that the players are self-interested and exert costly effort in order to win valuable prizes. Each player is associated with a private \textit{ability} (or \textit{quality}), and their cost, as a function of their output, is linear with a slope equal to the inverse of their ability. The players know the prize allocation scheme, their own ability, and the prior distributions of other players' abilities, and play strategically, reaching a Bayes--Nash equilibrium. On the other hand, the contest designer knows the prior distributions of the players' abilities, and can therefore compute the equilibrium behavior of the players. She wants to design the prize allocation scheme to elicit equilibrium behavior that optimizes her own objective.

\subsection{Designer's Objective}
The most widely studied designer's objective in the literature is the total output, i.e., the sum of the outputs generated by the players. 
Under the total output objective, the designer values equally the marginal output by weak players (producing low output) and strong players (producing high output). In particular, if there are $n$ players in the contest and $b_i$ is the output of the $i$-th player, then the total output objective is
\[
    \text{total output: }\quad \max \sum_{i = 1}^n b_i.
\]
However, in several practical scenarios, the designer may want to focus on the output generated by a section of players producing low/middle/high level of output or to elicit an adequate output from several players instead of very high output from a few players. Motivated by this observation, we consider objectives of the following type
\[
    \text{our designer's objective: }\quad \max \sum_{i = 1}^n f(b_i),
\]
where the function $f$ takes the \textit{threshold} form shown in Figure~\ref{fig:objectives} and described below.

\begin{figure}
    \centering
    \includegraphics[width=.32\textwidth]{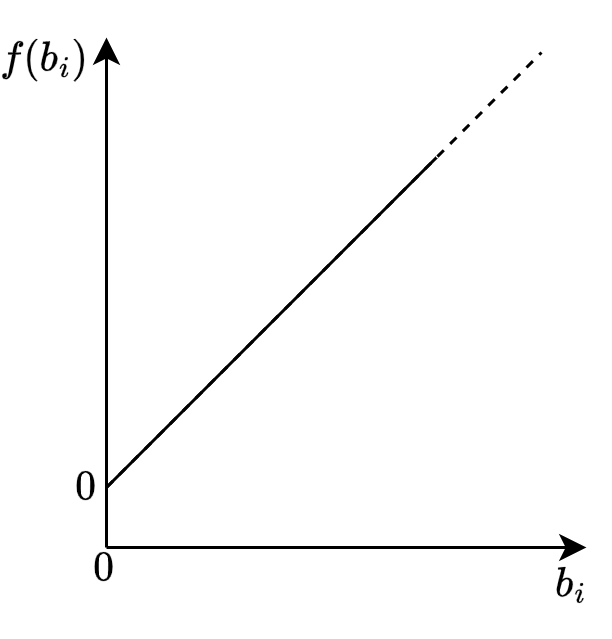}\hfill
    \includegraphics[width=.32\textwidth]{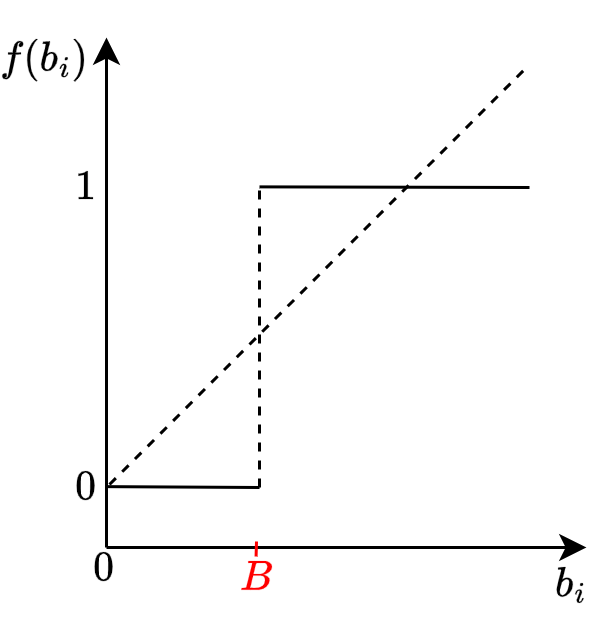}\hfill
    \includegraphics[width=.32\textwidth]{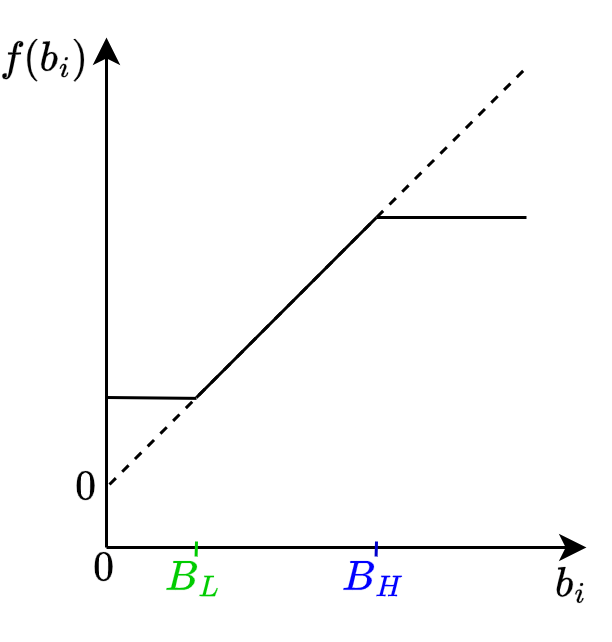}\hfill
    \caption{Designer's objectives (left to right): total output, binary threshold, and linear threshold.}
    \label{fig:objectives}
\end{figure}

We focus on two types of objective functions. Under the {\em binary threshold objective}, there is a fixed threshold $B$, and the designer's objective value weakly increases with the number of \textit{successful} players. A player is successful if she produces an output of at least $B$. In contrast, under the {\em linear threshold objective}, there is a lower threshold $B_L$ and an upper threshold $B_H$. All players who produce an output below $B_L$ make the same contribution to the designer's objective; similarly, all players who produce an output above $B_H$ contribute the same amount. However, between these thresholds, a player's output contributes linearly to the designer's objective, just like in the case of the total output objective.
These threshold objectives are approximations of an S-shaped curve; the similarity can be observed from the plots in Figure~\ref{fig:objectives}. An S-shaped objective, and the threshold objectives that approximate it, discount the marginal output from the tails and focus on the players in the middle.
In the appendix, we also consider convex and concave objectives, which correspond to the bottom- and top-halves of S-shaped objectives, respectively.

A designer with a binary threshold objective aims to elicit an adequate output from several players.
Consider, for instance, an instructor who is preparing the students for a standardized test that measures the school performance, but has a limited impact on the students' educational trajectories. The instructor/school may want to incentivize her students to perform better by giving rewards. If a student is far below the pass/fail threshold, they are likely to fail even if they improve their performance a bit; similarly, there is no need to push students who are sure to excel to work even harder. The crucial students are the ones in between, and the instructor wants to award the prizes to elicit additional output from these students.
Another example could be online competition and crowdsourcing platforms (e.g., Topcoder\footnote{\url{https://www.topcoder.com/}}) and hackathons. These competitions are sometimes organized to encourage learning and adoption of new tools and technologies, and the designer may prefer to get sincere participation from many players.
Similarly, a health insurance company promoting a fitness/exercise app aims to encourage many subscribers to start exercising regularly rather than to get a few fitness enthusiasts to log many hours each day. They could use their budget to give rewards to the active subscribers.

The binary threshold objective has a jump in the designer's utility at the threshold, but all players above it (or below it) are the same for the designer. This may be reasonably expressive for some applications, but for others, a smoother threshold region may be desirable. The linear threshold objective provides this added expressiveness. 
Slightly extending the motivating examples mentioned earlier, for the standardized school test, if the test is of average difficulty but provides relatively fine-grained grades, the excellent students who are sure to get a full score  (or the students at the other tail) are still less relevant, but the other students who can increase their score with additional effort are relevant.
Similarly, in the hackathon (and online competition) or the fitness app examples, the designer may benefit from a marginal increase in the output of moderately performing individuals.

The examples discussed above are likely to be best modeled using a general S-shaped objective. However, the threshold objectives can be good approximations of such S-shaped objectives for specific applications and are analytically tractable. Also, the threshold functions only have a few parameters---the thresholds---compared to a general S-shaped curve, which can be a desirable property for a practitioner.


\subsection{Other Modeling Choices}
We consider the following constraints that a designer may face while designing contests to optimize the threshold objectives. The most natural is rank-order allocation of prizes with unit-sum constraints, but we study the other cases as well to get a complete picture of the structure of the optimal contest.

\subsubsection{Rank-Order Allocation and General All-Pay Allocation}
A contest designer might be restricted to using only the relative value of the players' outputs to award prizes, which motivates the study of contests with a rank-order allocation of prizes (see, e.g.,~\cite{moldovanu2001optimal}). For example, the instructor may rank the students based on their performance and award the prizes based on the ranks.
On the other hand, the contest designer might be allowed to use the numerical values of the players' outputs and design the optimal general all-pay contest (see, e.g.,~\cite{chawla2019optimal}). For example, the instructor may decide that the students get prizes based on their ranks but only if they get above a minimum score in the test.

\subsubsection{Unit-Sum and Unit-Range}
We assume that the prizes are non-negative, and we normalize them in two ways: \textit{unit-sum} and \textit{unit-range}. 
The unit-sum constraint is a budget constraint, which requires that the total prize awarded does not exceed $1$ (normalized). For example, the instructor (or the hackathon organizer or the fitness app) may have a limited budget to reward the participating players.
The unit-range constraint restricts the individual prizes awarded to the players to be between $0$ and $1$ (normalized). 
Such a constraint is suitable when the designer is not restricted by a budget. For example, the instructor may award medals or provide letters of recommendation as a result of the test scores, and she might not have a rigid budget constraint. Similarly, the fitness app may provide reputation badges to incentivize agents.


\subsection{Our Results}
Overall, our results indicate that the contests that optimize the threshold objectives are strict generalizations of the contests that optimize the total output. For example, for the total output, the optimal contest allocates the entire prize budget (for unit-sum) to the top-ranked player, but for threshold objectives, lower-ranked players may also be awarded positive prizes. This may explain the widespread practice of awarding multiple prizes in a contest. On the other hand, the optimal contests are still simple, i.e., they are easy to interpret, compute, and implement. 

\subsubsection{Contests with Rank-Order Allocation}
We call a contest that equally distributes the prize among the top $k$ players (for some $k \le n$, where $n$ is the number of players) and gives $0$ to others a \textit{simple} contest. Any rank-order allocation contest can be written as a convex combination of these $n$ simple contests, and the equilibrium in the rank-order allocation contest is also the same convex combinations of the equilibria in the simple contests. This observation allows us to derive the following results.

For the binary threshold objective, the optimal contest is a simple contest, i.e., it equally distributes the prize among the top $k$ players. 
The value of $k$ depends upon the threshold $B$ and the distribution $F$ of the players' abilities but not upon the non-decreasing function, say $\rho : \{0\} \cup [n] \rightarrow \bR_{\ge 0}$, that maps the number of successful players to the utility of the designer (Theorem~\ref{thm:rankBinary}). The independence on $\rho$ follows from Lemma~\ref{thm:binaryEquiv}, which shows that the optimal contest for any non-decreasing $\rho$ is the same.
Intuitively, among all the simple contests, there is one that maximizes the probability that an arbitrary agent (with ability i.i.d. from $F$) produces output above the threshold $B$, and this contest is the optimal contest.


The linear threshold objective is more complex than the binary threshold. An optimal contest for this case, in addition to increasing the number of players producing output above the upper threshold, also cares about increasing the output of the players between the two thresholds. We show that the optimal contest is a convex combination of at most $3$ simple contests.
In other words, the optimal contest has up to three levels of prizes: the top $k$ players get the first-level prize, the next $\ell$ players get the second-level prize, the next $m$ players get the third-level prize, and the last $n-k-\ell-m$ players do not get anything; here, $n$ is the number of players, and the values of $k$, $\ell$, and $m$ depend upon the value of the thresholds and $F$ (Theorem~\ref{thm:rankLinear}). 
We also prove that a simple contest (recall that this format is optimal for the binary threshold) has an approximation ratio of $2$ for the linear threshold (Theorem~\ref{thm:rankSimpleVsOptimal}). 

For both objectives, the results apply to both unit-sum and unit-range constraints on the prizes. Qualitatively, the results are the same for these constraints, but the number of players at the different prize levels may vary slightly depending upon the type of constraint.

    
\subsubsection{General All-Pay Contests}
General all-pay contests can use the numerical values of the agents' outputs (in addition to the relative ordering, as used by rank-order allocation contests). The optimal contests strongly leverage this added expressive power.

For the binary threshold objective, the optimal contest is the most intuitive one. It equally distributes the prize to the players who produce an output above a reserve output level, and this reserve is equal to the threshold (Theorem~\ref{thm:optBinary}). This is true for both unit-sum and unit-range constraints, and is independent of the function that maps the number of successful players to the utility of the designer (Lemma~\ref{thm:binaryEquiv}).

For the linear threshold objective, the optimal contest is an extension of the revenue-maximizing all-pay auction with a reserve~\citep{myerson1981optimal}. 
For the unit-range constraint, the problem becomes a standard principal-agent problem (single agent) as the allocation for a given agent can be optimized independently of others.
If the distribution of the players' abilities, $F$, is \textit{regular}\footnote{See Definition~\ref{def:regular}. This is a weaker assumption than the \textit{monotone hazard rate} condition.}, then there is a reserve output between the lower and the upper threshold, and any player with an output above the reserve gets a prize of $1$ (Theorem~\ref{thm:optLinearReg}). 
For the unit-sum constraint and regular $F$, the contest has a reserve output and a \textit{saturation} output. In this case, the prize allocation depends on whether the player with the highest output is: 
(i) below the reserve, (ii) between the reserve and the saturation level, or (iii)
above the saturation level. In case (i), no one gets a prize; in case (ii), the player with the highest output gets the entire prize; and in case (iii), the prize is distributed equally among the players with outputs above the saturation level (Theorem~\ref{thm:optLinearReg}). The reserve and the saturation levels depend upon $F$. For \textit{irregular} $F$, following techniques from optimal auction design, we \textit{iron} the \textit{virtual ability function} to get an optimal contest that is a generalization of the optimal contest for the regular case (Theorem~\ref{thm:optLinearIr}).


\subsection{Related Work}
Rank-order allocation of prizes is the dominant paradigm in contest theory. Our work is closely related to prior work on contest design with incomplete information and unit-sum constraints~\citep{glazer1988optimal,moldovanu2001optimal,chawla2019optimal}.
Glazer and Hassin~\cite{glazer1988optimal} show that for linear cost functions and players' abilities sampled i.i.d. from a uniform distribution, 
the contest that maximizes the total output awards the entire prize to the top-ranked player. Moldovanu and Sela~\cite{moldovanu2001optimal} give the symmetric Bayes--Nash equilibrium (Theorem~\ref{thm:beta}) that we use in our analysis. They also generalize the result of Glazer and Hassin~\cite{glazer1988optimal} and show that awarding the entire prize to the top-ranked player is optimal when the players have (weakly) concave cost functions with the abilities sampled i.i.d. from any distribution with continuous density function; however, with convex cost functions, the optimal mechanism can have multiple prizes.%
\footnote{Specific cases of our problem (such as the linear threshold objective with only an upper threshold) are related, although not equivalent, to the problem of maximizing total output when players have non-linear cost functions. A non-linear cost function affects the players' equilibrium behavior, but a threshold objective is associated with the designer and does not directly affect the equilibrium behavior. For example, if every player has a convex cost function that is linear up to a certain threshold and then goes to infinity (this setting is similar in spirit to the linear threshold objective with an upper threshold), then no player would produce an output above the threshold, but we will see that the optimal rank-order allocation contest with a linear threshold objective may have players that produce an output above the upper threshold.}
Chawla et al.~\cite{chawla2019optimal} optimize maximum individual output instead of total output, in a similar incomplete information setup with linear cost functions. Here too the optimal contest allocates the entire prize to the top-ranked player. 

Our study of the optimal general all-pay contest design rests on the framework established in the seminal work by Myerson~\cite{myerson1981optimal} on revenue optimal auction design. DiPalantino and Vojnovic~\cite{dipalantino2009crowdsourcing} and Chawla et al.~\cite{chawla2019optimal} connect crowdsourcing contests with all-pay auctions. For general all-pay contests with linear cost functions, the optimal contest for total output has been studied by Vojnovic~\cite{vojnovic2015contest} and the optimal contest for maximum individual output has been considered by Chawla et al.~\cite{chawla2019optimal}. 
The optimal contests for both these objectives have a structure similar to Myerson's optimal auction: they allocate the total budget to the player with the highest output above a reserve output level for regular distributions (and the highest ironed output for irregular distributions).

Mechanisms that are similar to the binary threshold objective have been discussed in previous works~\citep{taylor1995digging,halac2017contests,loury1979market}. For example, Taylor~\cite{taylor1995digging} mentions ``in a research tournament, the terminal date is fixed, and the quality of innovations varies, while in an innovation race, the quality standard is fixed, and the date of discovery is variable.'' In the innovation race~\citep{loury1979market,taylor1995digging}, the objective is to get at least one very good outcome, and a winning criterion that corresponds to a binary threshold (with ties broken in favor of the player who first reaches this threshold) is a proposed mechanism. In our work, the binary threshold is the objective (and not the mechanism), and our aim is to find the optimal mechanism to maximize this objective. Also, even if we interpret the innovation race's mechanism as an objective, it would be a specific instance of the binary threshold objective where the designer's utility is $1$ if there is at least one successful player and $0$ otherwise (rather than an arbitrary increasing function of the number of successful players).

To the best of our knowledge, there has not been any work on maximizing the linear threshold objectives studied in this paper. In addition to maximizing the total output (e.g.,~\cite{glazer1988optimal,moldovanu2001optimal,moldovanu2006contest,minor2011increasing}) and the maximum individual output (e.g.,~\cite{chawla2019optimal,moldovanu2001optimal,taylor1995digging,ales2017optimal,mihm2019sourcing}), other objectives that have been investigated include maximizing the cumulative output from the top $k$ agents (e.g.,~\cite{archak2009optimal,gavious2014revenue}).

On the technical side, our work on rank-order allocation builds upon the equilibrium characterization of Moldovanu and Sela~\cite{moldovanu2001optimal}. As in the prior work, the \textit{single-crossing} property (Definition~\ref{def:singleCrossing}) and properties of order statistics are useful for the characterization of the optimal rank-order allocation contest. For general all-pay contests, we build upon the work on revenue optimal auction design by Myerson~\cite{myerson1981optimal} and on the implementability of auctions by Matthews~\cite{matthews1984implementability}. Matthews~\cite{matthews1984implementability} characterizes which interim allocation functions can be implemented by some allocation function, and therefore allows us to focus on interim allocation functions instead of allocation functions. Previous works in general contest design, such as the work of Chawla et al.~\cite{chawla2019optimal}, did not require the result of Matthews~\cite{matthews1984implementability} because, unlike in our model, their objective functions were linear in the interim allocation.

There have also been several studies in the complete information settings (e.g.,~\cite{baye1996all,barut1998symmetric}). We point the readers to the book by Vojnovic~\cite{vojnovic2015contest}, particularly Chapter~3, for a survey on related topics.

\section{Notation and Preliminaries}\label{sec:prelim}
For an integer $k \in \bZ_{\ge 0}$, let $[k] = \{1,2, \ldots, k\}$.

There are $n$ players. Let $\bv = (v_1, v_2, \ldots, v_n)$ be the ability profile of the players, where the values $v_i$ are drawn independently from a continuous and differentiable distribution $F$ with support $[0,1]$. Let $f$ be the probability density function (PDF) of $F$. The $n$ players simultaneously produce outputs $\bb = (b_1, b_2, \ldots, b_n) \in \bR^n_+$. Player $i$ has a cost of $b_i/v_i$ for producing output $b_i$; we shall use a scaled version of this cost for convenience, as given in \eqref{eq:rankPlayerUtility}. 
For a function $g(x)$ that is not one-to-one, let $g^{-1}(y)$ denote the minimum value of $x$ such that $g(x) \ge y$, for $y$ in the range of $g$.

\subsection{Contests with a Rank-Order Allocation of Prizes}
The contest has $n$ prizes $\bw = (w_1, w_2, \ldots, w_n)$, where $0 \le w_{j+1} \le w_j \le 1$ for $j \in [n-1]$. For the unit-sum model, we additionally require that $\sum_j w_j \le 1$. The prize $w_1$ is awarded to the highest-performing player, $w_2$ to the second highest, and so on; a player receives one of the prizes based on the rank of their outputs, with ties broken uniformly. 
Fix an output vector $\bb = (b_i)_{i \in [n]}$ and suppose that $b_{i_1} \ge b_{i_2} \ge \ldots \ge b_{i_n}$. Then the utility of player $i$ (scaled up by $v_i$ for convenience) is given by:
\begin{equation}\label{eq:rankPlayerUtility}
   u(v_i, \bb) = v_i \sum_{j \in [n]} w_j \frac{ \b1\{b_i = b_{i_j}\} }{| \{ k \mid b_k = b_{i_j} \} |}  - b_i,
\end{equation}
where $\b1$ is the indicator random variable. To interpret this formula, observe that
$\frac{ \b1\{b_i = b_{i_j}\} }{| \{ k \mid b_k = b_{i_j} \} |}$ is the probability that player $i$ receives the $j$-th prize, whereas
$\frac{b_i}{v_i}$ is the cost of producing output $b_i$ for player $i$.

Let $p_j(v)$ be the probability that a value $v \in [0, 1]$ is the $j$th highest among $n$ i.i.d. samples from $F$, given by the expression:
\begin{equation}
    p_j(v) = \binom{n-1}{j-1} F(v)^{n-j}(1-F(v))^{j-1}.
\end{equation}
Let $f_{n,j}$ be the PDF of the $j$-th highest order statistic out of $n$ i.i.d. samples from $F$, given by the expression:
\begin{equation}
    f_{n,j}(v) = \frac{n!}{(j-1)!(n-j)!}F(v)^{n-j}(1-F(v))^{j-1}f(v).
\end{equation}
A key role in the symmetric Bayes--Nash equilibrium of the contest is played by the order statistics of the abilities of the players.

Moldovanu and Sela~\cite{moldovanu2001optimal} characterize the unique symmetric Bayes--Nash equilibrium in rank-order allocation contests. Chawla and Hartline~\cite{chawla2013auctions} prove the uniqueness of this equilibrium in general (see also \cite{chawla2019optimal} for more details).

\begin{theorem}\citep{moldovanu2001optimal,chawla2013auctions}\label{thm:beta}
Consider the game that models the rank-order allocation contest with the values of placement prizes $w_1 \ge w_2 \ge \ldots \ge w_n \ge 0$. The unique Bayes--Nash equilibrium is given by
\begin{equation}\label{eq:rankOutput}
    \beta(v) = \sum_{j \in [n]} w_j \int_0^v t p'_j(t) dt,
\end{equation}
where $\beta(v)$ is the output generated by a player with ability $v$. 
\end{theorem}
\noindent Note that $\beta$ depends on the prize vector $\bw$, but we are suppressing it to keep the notation cleaner.
\begin{definition}[Simple Contest]
A rank-order allocation contest is called {\em simple} if there exists a $j\in [n]$ such that it gives a positive prize of equal value to the first $j$ players and $0$ to the other $n-j$ players.
\end{definition}

\begin{definition}[Single-Crossing]\cite{moldovanu2006contest}, \cite[Chapter 3]{vojnovic2015contest}\label{def:singleCrossing}
A function $f : [a,b] \rightarrow \bR$ is {\em single-crossing with respect to a function $g : [a,b] \rightarrow \bR$} if there exists a point $x^* \in [a,b]$ such that $f(x) \le g(x)$ for all $x \le x^*$ and $f(x) > g(x)$ for all $x > x^*$; when this is the case, we will also say that $f$ is single-crossing with respect to $g$ at $x^*$.
\end{definition}




\subsection{General All-Pay Contests}
We utilize the connection between contests and all-pay auctions made by \citep{dipalantino2009crowdsourcing,chawla2019optimal}. Leveraging the \textit{revelation principle}~\citep{myerson1981optimal}, for most of the analysis of general all-pay contests, we shall restrict our attention to direct revelation mechanisms and optimize over allocation rules $\bx(\bv) = (x_i(\bv))_{i \in [n]}$ that determine the allocation based on the abilities of the players $\bv = (v_i)_{i \in [n]}$. We interpret $x_i(\bv)$ as the expected value of the prize obtained by player $i$ given the ability profile $\bv$, where the expectation is taken over the choices of the tie-breaking mechanism. We recognize that, when running a contest, we do not have direct access to players' abilities; rather, the players must produce outputs, and the allocation function must be based upon the observed outputs. After deriving the optimal allocation function based on the players' abilities, we shall convert it to an allocation function based on the players' outputs.

For both unit-range and unit-sum settings, we have the restriction $0 \le x_i(\bv) \le 1$ for all $i\in [n]$. For unit-sum, we additionally have $\sum_i x_i(\bv) \le 1$. The unit-range case is comparatively easier, because $x_i(\bv)$ can be optimized independently for every player $i$, whereas for unit-sum, for two players $j$ and $\ell$, $x_j(\bv)$ and $x_\ell(\bv)$ are not independent as we need to satisfy the $\sum_i x_i(\bv) \le 1$ constraint. 
While studying unit-sum contests, we shall by default assume that $n \ge 2$.

We assume that the allocation rule $\bx(\bv)$ is symmetric with respect to the players. As $\bx(\bv)$ is symmetric, we have $\bE[x_i(\bv) \mid v_i = v] = \bE[x_j(\bv) \mid v_j = v]$ for any $i,j \in [n]$. Let the interim (or expected) allocation function be $\xi(v) = \bE[x_i(\bv) \mid v_i = v]$. Using Myerson's  characterization of allocation rules that allow a Bayes--Nash equilibrium, we conclude that $\xi(v)$ should be non-negative and non-decreasing in $v$, and in the equilibrium, the output of a player as a function of her ability is given by:
\begin{equation}\label{eq:optOutput}
    \beta(v) = v \xi(v) - \int_0^v \xi(t)dt.
\end{equation}
We slightly abuse the notation by representing the output function as $\beta$ for both rank-order allocation contests and general all-pay contests.

Let us make a few crucial observations about $\bx(\bv)$ and $\xi$. Both our objective functions, binary and linear threshold~(formally defined in Definitions~\ref{def:bt} and \ref{def:lt}, respectively), depend upon the output function $\beta(v)$, which further depends upon the interim allocation function $\xi(v)$, but not directly on the allocation function $\bx(\bv)$. So, any allocation function $\bx(\bv)$ that leads to the same interim allocation function $\xi(v)$ leads to the same objective value. 

Observe that for any interim allocation function $\xi$ that is induced by an allocation function with unit-sum constraints, the following condition holds (check \cite{matthews1984implementability} for more details):
\begin{equation}\label{eq:optAvgAlloc1}
    \int_V^1 \xi(v) f(v) dv \le \frac{1-F(V)^n}{n}.
\end{equation}
In plain words, inequality~\eqref{eq:optAvgAlloc1} says that the probability that any player with an ability above $V$ gets a prize is at most the probability that any player has an ability above $V$. Now, we state a result in \cite{matthews1984implementability} that we shall use in the remainder of the paper.

\begin{theorem}\citep{matthews1984implementability}\label{thm:matthews}
    Any non-decreasing interim allocation function $\xi$ that satisfies inequality~(\ref{eq:optAvgAlloc1}) is implementable by some allocation function $\bx$ that satisfies unit-sum constraints.
\end{theorem}

Given this result, we can focus on finding a $\xi$ that is non-decreasing and satisfies inequality~\eqref{eq:optAvgAlloc1} without worrying about the unit-sum constraint $\sum_i x_i(\bv) \le 1$, because there will be some $\bx$ that implements $\xi$. Further, for the optimal $\xi$, as we shall see later, the optimal $\bx$ will be easy to derive.

In our study of general all-pay contests, we shall give special attention to regular distributions, defined below. This property of the distribution of the players' abilities allows for simpler and efficiently computable optimal contests (see \cite{myerson1981optimal} for their use in optimal auction design).

\begin{definition}[Regular Distributions]\label{def:regular}
A distribution $F$ is {\em regular} if $v(1-F(v))$ is concave with respect to $(1-F(v))$, or, equivalently, if the virtual ability $\psi(v) = v - \frac{1-F(v)}{f(v)}$ is non-decreasing in $v$.\footnote{In the auction theory literature, e.g., \cite{myerson1981optimal}, $v$ corresponds to the valuation and $\psi(v)$ to the virtual valuation, and, again, $F$ is regular if $\psi(v)$ is non-decreasing. Further, in auctions, $\psi(v)$ may be interpreted as the slope of the ``revenue curve'' at $v$, and more coarsely, the $v$ term in $\psi(v)$ as the maximum revenue obtainable and $\frac{1-F(v)}{f(v)}$ as the revenue loss due to not knowing $v$ in advance (a.k.a. ``information rent''). Similarly, in contests, $\psi(v)$ may be interpreted as the slope of the ``output curve'' at ability $v$.}
\end{definition}

\subsection{Objective Functions}
We formally define the binary threshold and linear threshold objective functions. They apply to both rank-order allocation contests and general all-pay contests. We use the same notation for the output function for both cases, $\beta$.

\begin{definition}[Binary Threshold Objective]\label{def:bt}
Under the binary threshold objective, the contest designer gets a utility of $\rho(k) \in \bR_{\ge 0}$ if there are $k \in \{0\} \cup [n]$ players with output equal to or above a specified threshold $B$. The expected utility of the designer is given by:
\begin{equation}
    BT = \bE_{\bv}[ \rho( \sum_{i \in [n]} \b1\{\beta(v_i) \ge B\} ) ].
\end{equation}
We assume that $\rho$ is non-decreasing, i.e., $\rho(k-1) \le \rho(k)$ for all $k \in [n]$, and normalize $\rho(0) = 0$.
\end{definition}
The next lemma, Lemma~\ref{thm:binaryEquiv}, proves that the optimal contest for the binary threshold objective (for both rank-order allocation and general all-pay contests) is the same for any non-decreasing $\rho$. Further, given Lemma~\ref{thm:binaryEquiv}, we can maximize $\bE_v[\b1\{\beta(v) \ge B\}]$ to get the optimal contest for the binary threshold objective.
\begin{lemma}\label{thm:binaryEquiv}
The contest that maximizes $\bE_v[\b1\{\beta(v) \ge B\}]$ maximizes the binary threshold objective $\bE_{\bv}[ \rho( \sum_{i \in [n]} \b1\{\beta(v_i) \ge B\} ) ]$ for any non-decreasing $\rho$.
\end{lemma}

\begin{definition}[Linear Threshold Objective]\label{def:lt}
Under the linear threshold objective, the contest designer's utility increases linearly with a player’s output if the player’s output is between a lower threshold of $B_L$ and an upper threshold of $B_H$. Formally, we define it as:
\begin{multline}
    LT = \bE_{\bv}[\sum_{i \in [n]} \max(0, \min(B_H,\beta(v_i))-B_L)] \\ 
    = n \bE_{v}[\max(0, \min(B_H,\beta(v))-B_L)] = n \bE_{v}[\max(B_L, \min(B_H,\beta(v)))] - nB_L \\
    \equiv \bE_{v}[\max(B_L, \min(B_H,\beta(v)))] = \int_0^1 \max(B_L, \min(B_H,\beta(v))) f(v) dv,
\end{multline}
where equivalence above denotes the fact that maximizing the left hand side is equivalent to maximizing the right hand side.
\end{definition}

Note that the thresholds---$B$ for the binary threshold objective and $B_L$ and $B_H$ for the linear threshold objective---are exogenous.

\section{Rank-Order Allocation of Prizes}\label{sec:rank}
In this section, we study contests that allocate prizes based on the players' ranks. We first present some useful properties of the equilibrium output function $\beta$ given in Theorem~\ref{thm:beta}. Then we study the two objective functions based on these properties.

From Theorem~\ref{thm:beta}, we have
\begin{equation*}
    \beta(v) = \sum_{j \in [n]} w_j \int_0^v t p'_j(t) dt.
\end{equation*}
Writing $p'_j(t)$ using order statistics:
\begin{align*}
    p'_j(t) = 
    \begin{cases}
    f_{n-1,1}(t), &\text{ if } j = 1 \\
    f_{n-1,j}(t) - f_{n-1,j-1}(t), &\text{ if } 1 < j < n \\
    -f_{n-1,n-1}(t), &\text{ if } j = n
    \end{cases}.
\end{align*}
Substituting $p'_j(t)$ into the formula for $\beta(v)$, we get
\begin{equation}\label{eq:beta}
    \beta(v) = \sum_{j \in [n-1]} (w_j - w_{j+1}) \int_0^v t f_{n-1,j}(t) dt.
\end{equation}
From equation~\eqref{eq:beta} we can observe that decreasing $w_n$ to $0$ does not decrease $\beta(v)$ for any $v$. Changing $\beta(v)$ so that it becomes greater or higher for every $v$ leads to an equal or higher utility for the designer for both binary and linear threshold objectives. So, from now on, we shall assume that $w_n = 0$.

Depending upon whether we are looking at the unit-range or the unit-sum constraint on prizes, we have different constraints on $\bw$. We now transform $\beta(v)$ further to make it more convenient to work with.

\paragraph{Unit-Sum} We have the constraints $\sum_j w_j \le 1$ and $w_j \ge w_{j+1} \ge 0$. Let $\alpha_j = j (w_j - w_{j+1})$. The set of constraints on $\balpha = (\alpha_j)_{j \in [n-1]}$ that are equivalent to the constraints on $\bw$ are: $\alpha_j \ge 0$ for all $j \in [n-1]$ and $\sum_{j \in [n-1]} \alpha_j \le 1$. Let $\beta_S$ denote the output function $\beta$ with unit-sum constraints. We can rewrite equation~\eqref{eq:beta} as
\begin{equation}\label{eq:betaS}
    \beta_S(v) = \sum_{j \in [n-1]} \alpha_j \frac{1}{j} \int_0^v t f_{n-1,j}(t) dt = \sum_{j \in [n-1]} \alpha_j \beta_{Sj}(v),
\end{equation}
where $\beta_{Sj}(v) := \frac{1}{j} \int_0^v t f_{n-1,j}(t)dt$. Observe that the simple contest that awards a prize of $1/j$ to the first $j$ players and $0$ to others has $\alpha_j = 1$ and $\alpha_k = 0$ for $k \neq j$. Moreover, this contest induces an output of $\beta_{Sj}(v)$ from a player with ability $v$. Thus, any rank-order prize structure can be written as a convex combination of these $(n-1)$ simple contests where the first $j$ players get awarded $1/j$, for $j \in [n-1]$. 

\paragraph{Unit-Range} We have the constraints $1 \ge w_j \ge w_{j+1} \ge 0$ for all $j \in [n-1]$. In this case, let $\alpha_j = w_j - w_{j+1}$. We have the same set of  constraints on $\balpha$ as with unit-sum: $\alpha_j \ge 0$ for all $j$ and $\sum_{j \in [n-1]} \alpha_j \le 1$. However, we have a slightly different formula for $\beta$, which we denote by $\beta_R$:
\begin{equation}\label{eq:betaR}
    \beta_R(v) = \sum_{j \in [n-1]} \alpha_j \int_0^v t f_{n-1,j}(t) dt = \sum_{j \in [n-1]} \alpha_j \beta_{Rj}(v),
\end{equation}
where $\beta_{Rj}(v) := \int_0^v t f_{n-1,j}(t) dt$. Thus, similarly to the unit-sum case, any unit-range contest and the respective $\beta_R(v)$ can be written as a convex combination of $(n-1)$ simple unit-range contests $\beta_{Rj}$, $j \in [n-1]$. However, the unit-range contest that induces $\beta_{Rj}$ awards a prize of $1$ to the first $j$ players and $0$ to others, whereas the unit-sum contest $\beta_{Sj}$ awards a prize of $1/j$ to the first $j$ players.

Using the characterization of $\beta$ in equation~\eqref{eq:betaS} for unit-sum and~\eqref{eq:betaR} for unit-range, we can easily prove that the \textit{total output} objective is maximized by $\beta_{S1}$ for unit-sum contests (proved by \cite{glazer1988optimal,moldovanu2001optimal}) and by a simple contest for unit-range contests (the proof is provided in Appendix~\ref{sec:app:totalOutput}). 

Most of our analysis in this section applies to both unit-range and unit-sum settings; we shall use $\beta$ to denote either $\beta_S$ or $\beta_R$, and $\beta_j$ to denote either $\beta_{Sj}$ or $\beta_{Rj}$. Also, we shall assume without loss of generality that $\sum_j \alpha_j = 1$, because increasing $\alpha_i$ for some $i \in [n-1]$ while keeping $\alpha_j$ constant for all $j \in [n-1] \setminus \{i\}$ does not decrease $\beta(v)$ for any $v \in [0,1]$, and therefore does not decrease either of our two objective functions.

\begin{theorem}\label{thm:singleCrossing5}
Fix an $\balpha$ and the corresponding output function $\beta$. Consider a pair of indices $j,k$ s.t. $1 \le j < k \le n-1$, and $\epsilon > 0$. 
Suppose both the vector $\balpha$ and the vector $\balpha'$ given by $\alpha'_j = \alpha_j + \epsilon$, $\alpha'_k = \alpha_k - \epsilon$, $\alpha'_\ell = \alpha_\ell$ for $\ell \notin \{j,k\}$ satisfy the required constraints. Let $\beta'$ be the output function that corresponds to $\balpha'$. Then, $\beta'$ is single-crossing w.r.t. $\beta$, and $\beta^{-1}$ is single-crossing w.r.t. $\beta'^{-1}$.
\end{theorem}
The proof of Theorem~\ref{thm:singleCrossing5} for unit-sum is available in \cite[Chapter 3]{vojnovic2015contest}; in Appendix~\ref{sec:app:rank}, for completeness, we provide a similar proof for both unit-sum and unit-range settings.

\subsection{Binary Threshold Objective}
We first focus on the binary threshold objective (Definition~\ref{def:bt}; Lemma~\ref{thm:binaryEquiv}): $\bE_v[\b1\{\beta(v) \ge B\}]$. 

\begin{theorem}\label{thm:rankBinary}
The rank-order allocation contest that optimizes the binary threshold objective is simple, and the output function for the optimal contest is $\beta_{j^*}$, where $j^*$ is selected from the set
\[ \argmin_j \beta_j^{-1}(B). \]
\end{theorem}
Given Theorem~\ref{thm:rankBinary}, we can design the optimal contest by first finding the root of the equation $\beta_j(v) - B = 0$ for each $j \in [n-1]$. We can do this efficiently using a root-finding algorithm such as the bisection method because $\beta_j$ is continuous and monotone. Then, we select a $j$ with the smallest $\beta_j^{-1}(B)$.




\subsection{Linear Threshold Objective}\label{sec:rankLinear}
Next, we consider the linear threshold objective:
$\bE [\max(B_L, \min(B_H,\beta(v))) ]$ (Definition~\ref{def:lt}). For the binary threshold objective, the optimal contest was simple, but for the linear threshold this is not true in general. The following example illustrates this:


\begin{example}\label{ex:rank1}
Consider a contest with: 
unit-sum prizes; 
three players, $n=3$; 
uniform distribution, $F(v) = v$, $f(v) = 1$; 
lower threshold $B_L = 0.01$; 
upper threshold $B_H = 0.15$. 
The output function is $\beta(v) = \alpha_1 \beta_1(v) + \alpha_2 \beta_2(v)$ where 
$\alpha_1 + \alpha_2 = 1$ and 
$\beta_1(v) = \int_0^v t f_{2,1}(t) dt = \int_0^v 2 t^2 dt = \frac{2}{3} v^3$ and 
$\beta_2(v) = \frac12 \int_0^v t f_{2,2}(t) dt = \int_0^v t(1-t) dt = \frac{v^2}{2} - \frac{v^3}{3}$. 
We consider three contests: the two simple contests and a mixed one.
\begin{itemize}
    \item Simple Contest 1: $\alpha_1 = 1$ and $\beta(v) = \beta_1(v)$. 
    We have $\beta_1^{-1}(B_L) \approx 0.2467$ and $\beta_1^{-1}(B_H) \approx 0.6082$.
    The objective value is $\approx 0.08342$.
    
    \item Simple Contest 2: $\alpha_2 = 1$ and $\beta(v) = \beta_2(v)$. 
    We have $\beta_1^{-1}(B_L) \approx 0.1490$ and $\beta_1^{-1}(B_H) \approx 0.8042$.
    The objective value is $\approx 0.08218$.

    \item Mixed Contest: $\alpha_1 = \alpha_2 = 1/2$ and $\beta(v) = \beta_1(v)/2 + \beta_2(v)/2$.
    We have $\beta_1^{-1}(B_L) \approx 0.1885$ and $\beta_1^{-1}(B_H) \approx 0.6474$.
    The objective value is $\approx 0.08409$.
    
\end{itemize}
We observe that the given mixed contest outperforms the two simple contests.
\end{example}

For the case where there is only an upper threshold, i.e., $B_L = 0$, there is an optimal contest that is a convex combination of only two simple contests.  

\begin{theorem}\label{thm:rankLinearUp}
For a linear threshold objective with upper threshold only, i.e., with $B_L = 0$, there is an optimal $\balpha$ with at most two positive entries $\alpha_i$ and $\alpha_j$, i.e., with $\alpha_k=0$ for $k \in [n-1] \setminus \{i,j\}$. For this $\balpha$, $i$ and $j$, we also have:
\[ \int_0^{V_H} \beta_i(v) f(v) dv = \int_0^{V_H} \beta_j(v) f(v) dv = \int_0^{V_H} \beta(v) f(v) dv,\]
where $V_H = \beta^{-1}(B_H)$ and $\beta$ is the output function induced by $\balpha$.
\end{theorem}

Theorem~\ref{thm:rankLinearUp}  suggests an algorithm for finding the optimal $\balpha$, sketched below:
\begin{enumerate}
    \item \sloppy For every $i,j \in [n-1]$, $i < j$, find a $V_{ij} > 0$ (if any) such that $\int_0^{V_{ij}} \beta_i(v) f(v) dv = \int_0^{V_{ij}} \beta_j(v) f(v) dv$. Note that there might be multiple such values for $V_{ij}$, but these values form an interval of $[0,1]$ because $\int_0^{v} \beta_i(v) f(v) dv$ is single-crossing w.r.t. $\int_0^{v} \beta_j(v) f(v) dv$. Select any one of those values as $V_{ij}$. 
    \item If $V_{ij}$ exists and $\beta_i(V_{ij}) \ge B_H \ge \beta_j(V_{ij})$ then this pair $i,j$ is a candidate for being the optimal.
    \item The objective value for this pair is $\int_0^{V_{ij}} \beta_i(v) f(v) dv + B_H(1-F(V_{ij}))$. 
    \item Comparing $O(n^2)$ such pairs along with the $O(n)$ simple contests, we find the optimal contest.
    \item $\balpha$ corresponding to pair $i,j$ can be calculated as $\alpha_i = \frac{B_H - \beta_j(V_{ij})}{\beta_i(V_{ij}) - \beta_j(V_{ij})}$, $\alpha_j = 1 - \alpha_i$, and $\alpha_k = 0$ for $k \in [n-1] \setminus \{i,j\}$. (Check the proof of Theorem~\ref{thm:rankLinearUp} for more details about this step.)
\end{enumerate}

We can prove a result analogous to Theorem~\ref{thm:rankLinearUp} if we only have a lower threshold and no upper threshold (i.e., $B_H = 1$): there is an optimal contest that is a convex combination of at most two simple contests. Now, we give a result that applies for arbitrary thresholds.

\begin{theorem}\label{thm:rankLinear}
For a linear threshold objective, there is an optimal $\balpha$ with at most three positive entries $\alpha_i$, $\alpha_j$, and $\alpha_k$, i.e., with $\alpha_\ell=0$ for $\ell \in [n-1] \setminus \{i,j,k\}$. For this $\balpha$ and $i,j,k$, we also have:
\[ \int_{V_L}^{V_H} \beta_i(v) f(v) dv = \int_{V_L}^{V_H} \beta_j(v) f(v) dv = \int_{V_L}^{V_H} \beta_k(v) f(v) dv = \int_{V_L}^{V_H} \beta(v) f(v) dv, \]
where $V_L = \beta^{-1}(B_L)$, $V_H = \beta^{-1}(B_H)$, and $\beta$ is the output function induced by $\balpha$.
\end{theorem}
The main ingredients used in the proof of Theorem~\ref{thm:rankLinear} (and Theorem~\ref{thm:rankLinearUp}) are: first-order optimality condition of the objective w.r.t. $\balpha$; single-crossing property of $\beta_i$ w.r.t. $\beta_j$ for $i < j$; and the fact that every linear programming problem has an optimal solution that lies at a corner of the feasible region.

Recall that for the case $B_L=0$, we used Theorem~\ref{thm:rankLinearUp} to obtain an algorithm that compares at most $O(n^2)$ contests to find an optimal one. In a similar spirit, for general $B_L$ and $B_H$, we can use Theorem~\ref{thm:rankLinear} to obtain an algorithm that finds an optimal contest by comparing at most $O(n^3)$ contests.

Example~\ref{ex:rank2} provides the optimal rank-order allocation contest for a instance given in Example~\ref{ex:rank1}.
\begin{example}\label{ex:rank2}
We have the same scenario as Example~\ref{ex:rank1}:
unit-sum prizes; 
three players, $n=3$; 
uniform distribution, $F(v) = v$, $f(v) = 1$; 
lower threshold $B_L = 0.01$; 
upper threshold $B_H = 0.15$. 
The output function is $\beta(v) = \alpha_1 \beta_1(v) + \alpha_2 \beta_2(v)$ where $\alpha_1 + \alpha_2 = 1$ and $\beta_1(v) = \int_0^v t f_{2,1}(t) dt = \int_0^v 2 t^2 dt = \frac{2}{3} v^3$ and $\beta_2(v) = \frac12 \int_0^v t f_{2,2}(t) dt = \int_0^v t(1-t) dt = \frac{v^2}{2} - \frac{v^3}{3}$. 

The two simple contests and the mixed contest described in Example~\ref{ex:rank1} have objective values $0.08342$, $0.08218$, and $0.08409$, respectively. 
The optimal contest has $\alpha_1 \approx 0.43$ and $\alpha_2 = 1 - \alpha_1 \approx 0.57$, so $\beta(v) \approx 0.43\beta_1(v) + 0.57\beta_2(v)$, 
which gives $\beta_1^{-1}(B_L) \approx 0.1818$ and $\beta_1^{-1}(B_H) \approx 0.6561$.
And the objective value is $\approx 0.08411$.
\end{example}

\subsubsection{Simple vs. Optimal} 
In Theorem~\ref{thm:rankLinear}, we proved that a convex combination of at most three simple contests is optimal. We now compare this optimal contest with the best simple contest.

\begin{theorem}\label{thm:rankSimpleVsOptimal}
For the linear threshold objective, the objective value of the optimal contest is at most $2$ times that of the best simple contest.
\end{theorem}
In the proof of Theorem~\ref{thm:rankSimpleVsOptimal}, we use the expression given in Theorem~\ref{thm:rankLinear} and the single-crossing property of $\beta_i$ w.r.t. $\beta_j$ for any $i < j$ to show that the objective value of the optimal contest is at most the sum of the objective values of two simple contests. So, one of these two simple contests gives us an objective value at least half of the optimal.


\section{General All-Pay Contests}\label{sec:opt}
In the previous section, we restricted our focus to contests that awarded prizes based on players' ranks only. In this section, we relax this restriction and consider contests that may use the numerical values of the players' outputs to award prizes. 

As we discussed in the preliminaries (Section~\ref{sec:prelim}), for the unit-range constraint, the allocation function must satisfy $0 \le x_i(\bv) \le 1$. Equivalently, the expected allocation function $\xi$ must satisfy $0 \le \xi(v) \le 1$. For the unit-sum constraint, in addition to the constraints for the unit-range case, the allocation function must also satisfy $\sum_i x_i(\bv) \le 1$. Equivalently (by Theorem~\ref{thm:matthews}), the expected allocation function $\xi$ must satisfy inequality~\eqref{eq:optAvgAlloc1}: $\int_V^1 \xi(v) f(v) dv \le \frac{1-F(V)^n}{n}$. From Myerson's characterization of allocation functions that allow a Bayes--Nash equilibrium, we know that $\xi$ must be monotonically non-decreasing.

The following allocation rules award the entire prize to players with abilities above $V$, if there are such players in a given ability profile. Therefore, they maximize $\int_V^1 \xi(v) f(v) dv$, and the inequality~\eqref{eq:optAvgAlloc1} is satisfied with an equality. 
\begin{enumerate}
    \item Give the prize to the player with the highest ability. Then the expected allocation function is $\xi(v) = F(v)^{n-1}$, and $\int_V^1 \xi(v) f(v) dv = \int_V^1 F(v)^{n-1} f(v) dv = \int_{F(V)}^1 y^{n-1} dy = \frac{1-F(V)^n}{n}$. We can also observe that for this allocation rule, inequality~\eqref{eq:optAvgAlloc1} is tight for every $V \in [0,1]$.
    \item \sloppy Uniformly distribute the prize among the players with abilities above $V$. Then the expected allocation function is $\xi(v) = \frac{1 - F(V)^n}{n(1 - F(V))}$, and $\int_V^1 \xi(v) f(v) dv = \int_V^1 \frac{1 - F(V)^n}{n(1 - F(V))} f(v) dv = \frac{1-F(V)^n}{n}$.
\end{enumerate}

\subsection{Binary Threshold Objective}
For the binary threshold objective (Definition~\ref{def:bt}; Lemma~\ref{thm:binaryEquiv}), we have: 
\begin{equation*}
    \max \bE[\b1\{\beta(v) \ge B\}] = \max \int_0^1 \b1\{\beta(v) \ge B\} f(v) dv \equiv  \min (\beta^{-1}(B)).
\end{equation*}
Thus, we want to find a $\xi$ that minimizes $\beta^{-1}(B)$.

\begin{theorem}\label{thm:optBinary}
The optimal general all-pay contest with the binary threshold objective gives a prize of $1$ to all players who produce an output above the threshold $B$ in the unit-range model and equally distributes the total prize of $1$ to all players who produce an output above the threshold $B$ in the unit-sum model.
\end{theorem}

\subsection{Linear Threshold Objective}
The linear threshold objective is $\bE[\max(B_L,\min(B_H,\beta(v)))]$ (Definition~\ref{def:lt}).


We focus on the case when $F$ is \textit{regular}; the extension to irregular distributions is given in Appendix~\ref{sec:app:irregular}. We shall see that the optimal contest resembles a highest bidder wins all-pay auction with a \textit{reserve bid} (if every player bids below this, no one gets the prize) and a \textit{saturation bid} (all bids above this level are considered equal), and the ties are broken uniformly. We shall also give an efficient way to compute the reserve and the saturation bids.
 
We can, w.l.o.g., make the following assumptions on the optimal expected allocation function:
\begin{lemma}\label{thm:optLinear1}
If $\xi$ is optimal, we can assume w.l.o.g. that $\xi(v) = 0$ for $v < \beta^{-1}(B_L)$.
\end{lemma}
\begin{lemma}\label{thm:optLinear2}
If $\xi$ is optimal, we can assume w.l.o.g. that $\xi(v) = \xi(V_H)$ for $v \ge \beta^{-1}(B_H)$, i.e., $\xi(v)$ is constant for $v \ge \beta^{-1}(B_H)$.
\end{lemma}
From now on, we shall assume that $\xi$ satisfies the assumptions formulated in these two lemmas. 

We can write our linear threshold objective, using the notation $V_L = \beta^{-1}(B_L)$ and $V_H = \beta^{-1}(B_H)$, as:
\begin{equation}\label{eq:optLinearObj}
     B_L F(V_L) + \int_{V_L}^{V_H} \beta(v) f(v) dv + B_H (1-F(V_H)).
\end{equation}
 
Intuitively, the following lemma says that we should push the area under the curve $\xi$ to the right, as much as possible. We prove this result for unit-sum; for unit-range it holds, as a corollary, if we set $n=1$.
\begin{lemma}\label{thm:optLinearRegRight}
Let $F$ be a regular distribution and suppose that the allocation function has to satisfy unit-sum constraints. Then there is an optimal $\xi$ such that $\int_v^1 \xi(t) f(t) dt = \frac{1}{n}(1-F(v)^n)$ for all $v \in [0,1]$ where $\beta(v) < B_H$ and $\xi(v) > 0$.
\end{lemma}

The previous lemma effectively says that for $v \in [V_L,V_H)$ we have 
\begin{equation*}
    \int_v^1 \xi(t) f(t) dt = \frac{1-F(v)^n}{n}.
\end{equation*}
As both sides of the above equation are continuous, taking the limit $v \rightarrow V_H$, we have:
\[\int_{V_H}^1 \xi(t) f(t) dt = \frac{1-F(V_H)^n}{n}.\]
We already know that $\xi$ is $0$ for $v < V_L$ and constant for $v \ge V_H$. So, we get 
\begin{equation}\label{eq:expAllocationPartial}
    \xi(v) = \begin{cases}
    0, & \text{ if $ v < V_L$} \\
    F(v)^{n-1}, & \text{ if $ V_L \le v < V_H$} \\
    \frac{1-F(V_H)^n}{n(1-F(V_H))}, & \text{ if $v \ge V_H$}
    \end{cases}
\end{equation}

We now prove the following lemma, which says that for an optimal allocation rule, the output generated by a player never goes strictly above the upper threshold $B_H$ almost surely in $v \in [0,1]$.

\begin{lemma}\label{thm:optLinearRegUp}
If $\xi$ is an optimal expected allocation function, then the induced output function $\beta$ almost surely satisfies $\beta(v) \le B_H$ for $v \in [0,1]$. This result holds whether or not the distribution $F$ is regular, and both for unit-range and for unit sum constraints.
\end{lemma}

Combining the optimal expected allocation rule given in \eqref{eq:expAllocationPartial} with Lemma~\ref{thm:optLinearRegUp}, we get 
\[
\beta(V_H) = \frac{V_H(1-F(V_H)^n)}{n(1-F(V_H))} - \int_{V_L}^{V_H} F(v)^{n-1} dv = B_H.
\]
Note that the above formula is applicable only if there exists a $V_H$. It may be possible that $\beta(v)$ never reaches the threshold $B_H$, i.e., $\beta(v) < B_H$ for $v \in [0,1]$. The above formula also tells us that given $V_L$ (or $V_H$), $V_H$ (or $V_L$) can be calculated efficiently because, keeping $V_L$ (or $V_H$) fixed, $\beta(v)$ is continuous and monotone in $V_H$ (or $V_L$).

We now discuss how to efficiently compute the reserve $V_L$ and the saturation $V_H$. We shall use the following notation in the next theorem: $\eta(x) = \frac{1-x^n}{n(1-x)}$, $\psi_u(v) = v - \frac{F(u) - F(v)}{f(v)}$, $V_{\low}$ is the solution of $V_{\low} F(V_{\low})^{n-1} = B_L$, $V_{\mmid}$ is the solution of $\int_{V_{\mmid}}^1 F(v)^{n-1} dv = 1 - B_H$, and $V_{\uup}$ is the solution of $B_H = V_{up} \eta(F(V_{\uup}))$.

\begin{theorem}\label{thm:optLinearReg}
The contest that optimizes the linear threshold objective for regular distributions has the following allocation and expected allocation functions:
\begin{enumerate}
    \item For unit-range, the expected allocation function $\xi$ and the allocation function $\bx(\bv)$ are given as:
    \[ 
    \xi(v) = \begin{cases}
    0, \text{ if $v < V$ }\\
    1, \text{ if $v \ge V$ }
    \end{cases}
    \qquad x_i(\bv) = \begin{cases}
    0, \text{ if $v_i < V$ }\\
    1, \text{ if $v_i \ge V$ }
    \end{cases}
    \]
    where $V = \max(B_L,\min(B_H,\psi^{-1}(B_L)))$.
    
    \item For unit-sum, the optimal solution is given by one of the two cases below. Only one of the two cases is applicable for a given instance of the problem. The first case corresponds to the scenario when the expected output never reaches the upper threshold, while the second case captures the alternate scenario.
    \begin{enumerate}
        \item The expected allocation function $\xi$ is:
        \begin{equation}
            \xi(v) = \begin{cases}
                0, &\text{ if $v < V_L$}  \\
                F(v)^{n-1}, &\text{ if $v \ge V_L$}  
            \end{cases}
        \end{equation}
        and the allocation function $\bx(\bv)$ is:
        \begin{equation}\label{sol:linearRegSum}
            x_i(\bv) = \begin{cases}
            0, &\text{ if $\max_j(v_j) < V_L$ or $i \notin W$} \\
            1/|W|, &\text{ if $i \in W$ and $\max_j(v_j) \ge V_L$} 
            \end{cases}
        \end{equation}
        where $W = \{ k \mid v_k = \max_j(v_j) \}$ and $V_L = \min(V_{\mmid},\max(V_{\low}, \overline{V_L}))$; where $\overline{V_L}$ is the solution of $\frac{B_L}{F(\overline{V_L})^{n-1}}  - \psi(\overline{V_L}) = 0$.
        \item The expected allocation function $\xi$ is:
        \begin{equation}
            \xi(v) = \begin{cases}
                0, &\text{ if $v < V_L$}  \\
                F(v)^{n-1}, &\text{ if $V_L \le v < V_H$}  \\
                \eta(F(V_H)), &\text{ if $v \ge V_H$}  
            \end{cases}
        \end{equation}
        and the allocation function $\bx(\bv)$ is:
        \begin{equation}\label{sol:linearRegSumCase1}
            x_i(\bv) = \begin{cases}
            0, &\text{ if $\max_j(v_j) < V_L$ or $i \notin W$} \\
            1/|W|, &\text{ if $i \in W$ and $V_L \le \max_j(v_j) < V_H$} \\
            1/|\widehat{W}|, &\text{ if $i \in \widehat{W}$ and $\max_j(v_j) \ge V_H$}
            \end{cases}
        \end{equation}
        \sloppy
        where $W = \{ k \mid v_k = \max_j(v_j) \}$, $\widehat{W} = \{ k \mid v_k \ge V_H \}$, $V_L = \min(V_{\uup},\max(V_{\mmid}, \overline{V_L}))$ and $V_H = \min(1,\max(V_{\uup}, \overline{V_H}))$; where $\overline{V_L}$ and $\overline{V_H}$ are the solutions of equations $\frac{B_L}{F(\overline{V_L})^{n-1}} - \psi_{\overline{V_H}}(\overline{V_L}) = 0$ and $\overline{V_H} \eta(F(\overline{V_H})) - \int_{\overline{V_L}}^{\overline{V_H}} F(v)^{n-1} dv = B_H$.
    \end{enumerate}
\end{enumerate}
\end{theorem}
The values of $V_L$ and $V_H$ derived in the theorem above can be efficiently computed if the distribution $F$ is known, see the proof for more details. Also note that, although the contest may seem complicated, it is reasonably simple from a player's perspective. The optimal contest that maximizes the total output has a reserve~\citep{vojnovic2015contest}, here we have a saturation value in addition to the reserve. A player need not know how these reserve and saturation values are computed.

Implementing this mechanism in practice, i.e., finding the allocation as a function of the outputs of the players, is not difficult. 
Let us consider the unit-sum case where a $V_H$ exists, i.e., where $\beta$ reaches $B_H$ in Theorem~\ref{thm:optLinearReg}, given a player's output $b$, map it to $g(b)$ as given below (assume $\epsilon \rightarrow 0$):
\begin{align}\label{sol:linearRegSumOutput}
    g(b) = 
    \begin{cases}
    0, &\text{ if } b < \beta(V_L), \\
    b, &\text{ if } \beta(V_L) \le b \le \beta(V_H - \epsilon), \\
    \beta(V_H - \epsilon), &\text{ if } \beta(V_H - \epsilon) < b < \beta(V_H) = B_H, \\
    \beta(V_H) = B_H, &\text{ if } b \ge \beta(V_H) = B_H,
    \end{cases}
\end{align}
where $\beta(v)$ is computed based on $\xi(v)$ given in Theorem~\ref{thm:optLinearReg}.
One can then distribute the prize equally among the players who have the maximum positive $g(b)$.

In Example~\ref{ex:gen1} below, we consider the same instance as Examples~\ref{ex:rank1}~and~\ref{ex:rank2}, but we compute the optimal general all-pay contest instead of the optimal rank-order allocation contest.
\begin{example}\label{ex:gen1}
We have the same scenario as Example~\ref{ex:rank1}~and~\ref{ex:rank2}: unit-sum prizes; three players, $n=3$; uniform distribution, $F(v) = v$, $f(v) = 1$; thresholds $B_L = 0.01$ and $B_H = 0.15$. 

The optimal contest in this case has a $V_L \approx 0.129$ and $V_H \approx 0.335$. And the objective value of the designer is $\approx 0.102$; notice that this is much higher than the optimal rank-order allocation contest with an objective value of $\approx 0.084$.
\end{example}

\section{Conclusion}
In this paper, we looked at two natural and practically useful objectives for a contest designer and described optimal contests for the objectives. An interesting open problem is to find how well can a contest without a reserve and/or a saturation output, or a rank-order allocation contest, approximate the optimal general contest, possibly with an additional player; we may expect to get a result along the lines of~\citep{bulow1996auctions}. Another extension of this work would be to study other practically relevant objective functions for the designer, monotone transformations other than the threshold transformations we studied in this paper. Combining the objectives for the designer studied in this paper with non-linear utility and cost functions for the players is also an important research direction.

\bibliographystyle{alpha}
\bibliography{ref}

\appendix
\section{Appendix for Section~\ref{sec:prelim}}
\subsection{Expected Allocation} 
We use the expression given below in Section~\ref{sec:opt}. For completeness, we provide its derivation.

Consider a player $i$ with ability $v$, and fix $a$ and $b$ so that
$0 \le a \le v \le b \le 1$. For each $k \in [n-1]$, let $p(v,a,b,k)$ be the probability that all players other than player $i$ have ability at most $b$,  with $k$ of them having ability in $[a, b]$. We have:
\[ p(v,a,b,k) = \binom{n-1}{k} F(a)^{n-1-k} (F(b) - F(a))^k. \]
Now, if the player with ability $v$ is allocated $1/(k+1)$ with probability $p(v,a,b,k)$, then the expected prize allocation for the player is:
\begin{align*}
\sum_{k = 0}^{n-1} \frac{p(v,a,b,k)}{k+1} &= \sum_{k = 0}^{n-1} \binom{n-1}{k} \frac{F(a)^{n-1-k} (F(b) - F(a))^k}{k+1} \\
&= \sum_{k = 0}^{n-1} \binom{n}{k+1} \frac{F(a)^{n-1-k} (F(b) - F(a))^k}{n} \\
&= \sum_{k = 1}^{n} \binom{n}{k} \frac{F(a)^{n-k} (F(b) - F(a))^{k-1}}{n} \\
&= \frac{1}{n(F(b) - F(a))} \sum_{k = 1}^{n} \binom{n}{k} F(a)^{n-k} (F(b) - F(a))^k \\
&= \frac{1}{n(F(b) - F(a))} \left(\sum_{k = 0}^{n} \binom{n}{k} F(a)^{n-k} (F(b) - F(a))^k - F(a)^n\right) \\
&= \frac{1}{n(F(b) - F(a))} \left(F(b)^n - F(a)^n\right) = \frac{1}{n}\cdot \frac{F(b)^n - F(a)^n}{F(b) - F(a)}. 
\end{align*}

\subsection{Properties of Single-Crossing Functions}
The following are some well-known properties of single-crossing functions that we shall use to prove some of our results. (Proofs are given for the reader's convenience.)
\begin{lemma}\label{thm:singleCrossing1}
For two functions $f,g : [a,b] \rightarrow \bR$, if $f$ is single-crossing with respect to $g$, then $F(x) = \int_a^x f(x) dx$ is single-crossing with respect to $G(x) = \int_a^x g(x) dx$.
\end{lemma}
\begin{proof}
Let $y^* \in [a,b]$ be the maximum value such that $F(y) \le G(y)$ for all $y \le y^*$. Such a value exists because the set $\{y: F(y)\le G(y)\}$ is closed and $F(a) = G(a) = 0$.

Suppose that $f$ is single-crossing with respect to $g$ at $x^*$. Then we have $y^* \ge x^*$, because for every $x \le x^*$ we have $g(x) - f(x) \ge 0$, and therefore for every $y \le x^*$ we have $G(y) - F(y) = \int_a^y (g(x) - f(x)) dx \ge 0$. Now, for $y > x^*$, we know that $f(y) > g(y)$ so $F(y) - G(y)$ is strictly increasing. So, if $F(x)$ goes above $G(x)$ at some point in $[a,b]$, then it will remain that way.
\end{proof}

\begin{lemma}\label{thm:singleCrossing2}
For two continuous and strictly increasing functions $f,g : [a,b] \rightarrow \bR$ with $f(a) = g(a) = 0$, if $f$ is single-crossing with respect to $g$ at $x^*$, then $g^{-1}$ is single-crossing with respect to $f^{-1}$ (in the intersection of the domains of $f^{-1}$ and $g^{-1}$) at $y^* = f(x^*) = g(x^*)$.
\end{lemma}

\begin{proof}
Consider a point $y$ that belongs to the domains of $f^{-1}$ and $g^{-1}$. Suppose $y\le y^*$. Since $f$ and $g$ are continuous and strictly increasing, there exists a unique point $x$ with $y=f(x)$ and a unique point $z$ with $y=g(z)$. Further, $y\le y^*$ implies $x\le x^*$ and hence $f(x)\le g(x)$. It follows that $z\le x$, i.e., $g^{-1}(y) \le f^{-1}(y)$.
In a similar manner, we can prove that $g^{-1}(y) > f^{-1}(y)$ for each point $y > y^*$ in the intersection of the domains of $f^{-1}$ and $g^{-1}$.
\end{proof}

\subsection{Omitted Proofs}
\subsubsection{Proof of Lemma~\ref{thm:binaryEquiv}}
Remember that we are in an incomplete information model where the ability of every agent is picked independently from the distribution $F$, and where an agent $i$ only knows her ability $v_i$ but not the ability $v_j$ of any other agent $j \neq i$. Moreover, as we are looking at symmetric mechanisms, from equations \eqref{eq:rankOutput} and \eqref{eq:optOutput} we know that $\beta(v_i) = \beta(v_j)$ if $v_i = v_j$ for any agents $i$ and $j$, i.e., the agents use the same function to map their ability $v$ to the output $\beta(v)$.

Let $q = \bE_v[\b1\{\beta(v) \ge B\}] = \bP_{v}[\beta(v) \ge B]$ be the probability that any given agent has output above $B$. The probability that $k$ agents have output above $B$ is equal to $\binom{n}{k} q^k (1-q)^{n-k}$. The binary threshold objective can be written as
\begin{multline*}
    \bE_{\bv}[ \rho( \sum_{i \in [n]} \b1\{\beta(v_i) \ge B\} ) ] = \sum_{k \in [n]} \rho(k) \binom{n}{k} q^k (1-q)^{n-k} \\
    = \sum_{k \in [n]} (\rho(k) - \rho(k-1)) \sum_{\ell = k}^n \binom{n}{\ell} q^{\ell} (1-q)^{n-\ell} = \sum_{k \in [n]} (\rho(k) - \rho(k-1)) g_k(q),
\end{multline*}
where, let, $g_k(q) = \sum_{\ell = k}^n \binom{n}{\ell} q^{\ell} (1-q)^{n-\ell} $. As $\rho$ is non-decreasing, so $\rho(k) - \rho(k-1) \ge 0$ for every $k \in [n]$. Now, if we show that $g_k(q)$ is a non-decreasing function of $q$, then maximizing $q$ will maximize the objective, which shall complete the proof. Let us differentiate $g_k(q)$ w.r.t. $q$, we get
\begin{align*}
    \frac{d g_k(q)}{d q} &= \sum_{\ell = k}^n \binom{n}{\ell} \frac{d (q^{\ell} (1-q)^{n-\ell})}{d q} \\
    &= n q^{n-1} + \sum_{\ell = k}^{n-1} \binom{n}{\ell} \left(\ell q^{\ell-1} (1-q)^{n-\ell} - (n-\ell) q^{\ell} (1-q)^{n-\ell-1} \right) \\
    &= n q^{n-1} + n \sum_{\ell = k}^{n-1} \left( \binom{n-1}{\ell-1} q^{\ell-1} (1-q)^{n-\ell} - \binom{n-1}{\ell} q^{\ell} (1-q)^{n-\ell-1} \right) \\
    &= n \binom{n-1}{k-1} q^{k-1} (1-q)^{n-k} \ge 0 \text{ as $q \in [0,1]$.}
\end{align*}

\section{Appendix for Section~\ref{sec:rank}}
\subsection{Total Output}\label{sec:app:totalOutput}
If $B_L = 0$ and $B_H = 1$ in the linear threshold objective (Definition~\ref{def:lt}), then the objective becomes simply $\bE[\beta(v)]$, i.e., the well-studied objective of maximizing the total output. For the unit-sum case, it has been shown by \cite{glazer1988optimal} that the optimal allocation of prizes is to award the entire prize to the top-ranked player, i.e., set $\alpha_1 = 1$ and $\alpha_j = 0$ for $j > 1$. This can also be observed from the following analysis:
\begin{align*}
     \int_0^1 \beta_S(v) f(v) dv &=  \int_0^1 \left(\sum_{j \in [n-1]} \frac{\alpha_j}{j} \int_0^v xf_{n-1,j}(x) dx\right) f(v) dv \\
    &= \sum_{j \in [n-1]} \frac{\alpha_j}{j} \int_0^1 \left(\int_0^v xf_{n-1,j}(x) dx\right)  f(v) dv \\
    &= \sum_{j \in [n-1]} \frac{\alpha_j}{j} \int_0^1 \left(\int_x^1 f(v) dv\right) xf_{n-1,j}(x) dx  \\
    &= \sum_{j \in [n-1]} \alpha_j \int_0^1 \frac{x}{j} \left(1-F(x)\right) f_{n-1,j}(x) dx  \\
    &= \frac{1}{n} \sum_{j \in [n-1]} \alpha_j \int_0^1 x f_{n,j+1}(x) dx.
\end{align*}
As $\int_0^1 x f_{n,j+1}(x) dx$ is the expectation of the $(j+1)$-st highest value out of $n$ i.i.d. samples from $F$, we have  
$$
\int_0^1 x f_{n,j+1}(x) dx \ge \int_0^1 x f_{n,k+1}(x) dx \text{ for } j \le k.
$$
Therefore, the optimal output function is $\beta_{S1}$.

\noindent For the unit-range objective, following similar steps, we get:
\begin{align*}
    \int_0^1 \beta_R(v) f(v) dv &= \frac{1}{n} \sum_{j \in [n-1]} \alpha_j\cdot j\cdot\int_0^1 x f_{n,j+1}(x) dx. 
\end{align*}
While for unit-sum, $\int_0^1 x f_{n,j+1}(x) dx$ decreases as $j$ increases, for unit-range, $j\int_0^1 x f_{n,j+1}(x) dx$ does not necessarily decrease as $j$ increases. Nevertheless, we can see that the optimal contest is a simple contest and the output function in the optimal contest is $\beta_{Rj^*}$ where:
\[ j^* \in \argmax_{j \in [n-1]} \left(j\int_0^1 x f_{n,j+1}(x) dx\right). \]

\subsection{Concave and Convex Transformations}
Let $h: [0,1] \rightarrow [0,1]$ be a monotone non-decreasing differentiable function, either concave or convex. The contest designer's objective function is:
\begin{equation}\label{eq:conObj}
    \bE_{\bv}[\sum_{i \in [n]} h(\beta(v_i))] \equiv \bE_{v}[h(\beta(v))] = \int_0^1 h(\beta(v)) f(v) dv.
\end{equation}
In Section~\ref{sec:rank}, for both unit-sum and unit-range, we derived that $\beta(v) = \sum_{j \in [n-1]} \alpha_j \beta_j(v)$ where $\balpha = (\alpha_j)_{j \in [n-1]}$ s.t. $\sum_{j \in [n-1]} \alpha_j = 1$ and $\alpha_j \ge 0$ for $j \in [n-1]$. For unit-sum, $\beta_j(v) = \beta_{Sj}(v) = \frac{1}{j} \int_0^v t f_{n-1,j}(t)dt$ (equation~\eqref{eq:betaS}), and for unit-range, $\beta_j(v) = \beta_{Rj}(v) = \int_0^v t f_{n-1,j}(t)dt$ (equation~\eqref{eq:betaR}). We shall use these results in this section.

\begin{lemma}\label{thm:rankConvexConcaveLemma}
If $h$ is convex (concave), then the objective function of the contest designer is also convex (concave) w.r.t. $\balpha$.
\end{lemma}
\begin{proof}
Take $\gamma \in [0,1]$ and two vectors $\balpha^{(1)}$ and $\balpha^{(2)}$, and let $\balpha = \gamma \balpha^{(1)} + (1-\gamma) \balpha^{(2)}$. Let $OBJ(\balpha)$ denote the objective value corresponding to $\balpha$, similarly, $OBJ(\balpha^{(1)})$ and $OBJ(\balpha^{(2)})$.
\begin{align*}
    OBJ(\balpha) &= \int_0^1 h(\sum_{j \in [n-1]} \alpha_j \beta_j(v)) f(v) dv \\
    &= \int_0^1 h(\gamma \sum_{j \in [n-1]} \alpha^{(1)}_j \beta_j(v) + (1-\gamma) \sum_{j \in [n-1]} \alpha^{(1)}_j \beta_j(v)) f(v) dv.
\end{align*}
If $h$ is convex, then $h(\gamma x + (1-\gamma) y) \ge \gamma h(x) + (1-\gamma) h(y) $ for any $x$ and $y$, and therefore, we get $OBJ(\balpha) \ge \gamma OBJ(\balpha^{(1)}) + (1-\gamma) OBJ(\balpha^{(2)})$, so $OBJ$ is convex. Similarly, we can prove that if $h$ is concave, then $OBJ$ is concave.
\end{proof}

For convex $h$, the results are similar to the results for the total output objective (Section~\ref{sec:app:totalOutput}). For both unit-sum and unit-range, there is a simple contest that is also optimal. Moreover, for unit-sum, the simple contest that awards the entire prize to the first ranked player is also optimal. These results are similar to the results for a model where the players have concave cost functions, as studied by \cite{moldovanu2001optimal}. There is a similar correspondence for concave $h$ and convex cost functions.
\begin{theorem}\label{thm:rankConvex1}
If $h$ is convex, then there is a simple and optimal rank-order allocation contest.
\end{theorem}
\begin{proof}
From Lemma~\ref{thm:rankConvexConcaveLemma}, we know that if $h$ is convex then the objective function is also convex w.r.t. $\balpha$. Any $\balpha$ can be written as a convex combination of the corner points, where a corner point has $\alpha_j = 1$ for some $j \in [n-1]$ and $\alpha_k = 0$ for all $k \neq j$. So, the optimal value of the objective is achieved at some of the corner points (it may be achieved at other points also).
\end{proof}

\begin{theorem}
If $h$ is convex, and the prizes are unit-sum, then the simple contest that awards the entire prize to the first-ranked player is optimal.
\end{theorem}
\begin{proof}
From Theorem~\ref{thm:rankConvex1}, we know that there is an optimal simple contest. Let this simple contest have $\alpha_j = 1$ for some $j \in [n-1]$ and $\alpha_k = 0$ for all $k \neq j$, so $\beta(v) = \beta_{Sj}(v)$. Comparing the objective value of the contest where $\alpha_1 = 1$ with this contest, we have:
\begin{align*}
    \int_0^1 h(\beta_{S1}(v)) f(v) dv - \int_0^1 h(\beta_{Sj}(v)) f(v) dv &= \int_0^1 (h(\beta_{S1}(v)) - h(\beta_{Sj}(v))) f(v) dv \\
    \ge \int_0^1 h'(\beta_{Sj}(v)) (\beta_{S1}(v) &- \beta_{Sj}(v)) f(v) dv, \quad \text{ as $h$ is convex.}
\end{align*}
Using (i) $h'(\beta_{Sj}(v))$ is monotonically increasing and non-negative because $\beta_{Sj}(v)$ is monotonically increasing and $h$ is convex, (ii) $\beta_{S1}(v)$ is single-crossing w.r.t. $\beta_{Sj}(v)$, and (iii) $\int_0^1 (\beta_{S1}(v) - \beta_{Sj}(v)) f(v) dv \ge 0$, as proved while studying total output (Section~\ref{sec:app:totalOutput}), we have our required result.
\end{proof}

\begin{theorem}
If $h$ is concave, then the optimal rank-order allocation contest need not be simple, but can be efficiently found by solving a concave maximization problem.
\end{theorem}
\begin{proof}
The linear threshold objective with only an upper threshold is an example of a concave $h$ (we smooth the non-differentiable point), and for this objective, we already gave an example that shows that a simple contest may not be optimal (Section~\ref{sec:rankLinear}). From Lemma~\ref{thm:rankConvexConcaveLemma}, we know that the objective is concave if $h$ is concave, so the optimal contest can be solved efficiently by solving a concave maximization problem (equivalently a convex minimization problem).
\end{proof}

\subsection{Omitted Proofs}\label{sec:app:rank}
\subsubsection{Proof of Theorem~\ref{thm:singleCrossing5}}
First, let us prove that $t f_{n-1,j}(t)$ is single-crossing w.r.t. $t f_{n-1,k}(t)$ for $j < k$. We will argue that the following inequality is true for sufficiently large values of $t$:
\begin{align*}
    & t f_{n-1,j}(t) > t f_{n-1,k}(t) \\
    &\iff t \frac{(n-1)!}{(j-1)!(n-j-1)!}F(t)^{n-j-1}(1-F(t))^{j-1}f(t) \\
    & \qquad > t \frac{(n-1)!}{(k-1)!(n-k-1)!}F(t)^{n-k-1}(1-F(t))^{k-1}f(t) \\
    &\iff \frac{(k-1)!(n-k-1)!}{(j-1)!(n-j-1)!} > \left(\frac{1-F(t)}{F(t)}\right)^{k-j}.
\end{align*}
Indeed, for fixed values of $j$, $k$, and $n$, the left-hand side of the above inequality is a positive constant. As $t$ increases, the right-hand side monotonically decreases from $\infty$ to $0$. Thus, the above inequality is true for all $t$ above some $t^*$, and we have proved that $t f_{n-1,j}(t)$ is single-crossing w.r.t. $t f_{n-1,k}(t)$.

Following similar steps, we can prove that $t \frac{1}{j} f_{n-1,j}(t)$ is single-crossing w.r.t. $t \frac{1}{k} f_{n-1,k}(t)$ for $j < k$. Instead of the inequality $\frac{(k-1)!(n-k-1)!}{(j-1)!(n-j-1)!} > (\frac{1-F(t)}{F(t)})^{k-j}$ above, we have the inequality $\frac{k!(n-k-1)!}{j!(n-j-1)!} > (\frac{1-F(t)}{F(t)})^{k-j}$, but the subsequent argument applies.

Now, let us prove that $\beta$ is single-crossing w.r.t. $\beta'$. Let us assume that the following inequality is true:
\[ \beta'(v) > \beta(v). \]
For unit-range, we have:
\begin{align*}
    \beta'_R(v) &> \beta_R(v) \\
    \iff \sum_{l \in [n-1]} \alpha'_l \int_0^v xf_{n-1,l}(x) dx &> \sum_{l \in [n-1]} \alpha_l \int_0^v xf_{n-1,l}(x) dx \\
    \iff \int_0^v x f_{n-1,j}(x) dx &> \int_0^v x f_{n-1,k}(x) dx.
\end{align*}
Using Lemma~\ref{thm:singleCrossing1} and the fact that $t f_{n-1,j}(t)$ is single-crossing w.r.t. $t f_{n-1,k}(t)$ for $j < k$, we obtain that $\beta_R'$ is single-crossing w.r.t. $\beta_R$ for unit-range. Similarly, for unit-sum, we use Lemma~\ref{thm:singleCrossing1} together with the result that $t \frac{1}{j} f_{n-1,j}(t)$ is single-crossing w.r.t. $t \frac{1}{k} f_{n-1,k}(t)$ for $j < k$ to prove that $\beta_S'$ is single-crossing w.r.t. $\beta_S$.

As $\beta'$ is single-crossing w.r.t. $\beta$, using Lemma~\ref{thm:singleCrossing2}, we get that $\beta^{-1}$ is single-crossing w.r.t. $\beta'^{-1}$.

\subsubsection{Proof of Theorem~\ref{thm:rankBinary}}
We need to prove that for any threshold value $B$, the optimal $\balpha$ has $\alpha_j = 1$ for some $j$ and $\alpha_k = 0$ for $k \neq j$. In other words, $\beta = \beta_j$ for some $j \in [n-1]$.

From Definition~\ref{def:bt}, and equations (\ref{eq:betaS}) and (\ref{eq:betaR}), we have:
\begin{multline*}
    \int_0^1 \b1\{\beta(v) \ge B\} f(v) dv = \int_0^1 \b1\{v \ge \beta^{-1}(B)\} f(v) dv = \int_{\beta^{-1}(B)}^1 f(v) dv 
    = 1 - F(\beta^{-1}(B)).
\end{multline*}
Thus, to maximize the binary threshold objective, we need to minimize $F(\beta^{-1}(B))$, and as $F$ is non-decreasing, we need to minimize $\beta^{-1}(B)$. For any $\balpha$ we have $\alpha_j \ge 0$ and $\sum_j \alpha_j = 1$, and for every $j$, the $\beta_j$ functions are monotone. Therefore $\min_j \beta_j^{-1} (B) \le \beta^{-1}(B)$. Thus, the optimal contest is simple, and the output function for the optimal contest is $\beta_{j^*}$, where $j^*$ is:
\[ j^* = \argmin_j \beta_j^{-1}(B). \]

\subsubsection{Proof of Theorem~\ref{thm:rankLinearUp}}
Let us assume that $\balpha$ is optimal and $\beta$ is the induced output function. Let $V_H = \beta^{-1}(B_H)$. Let $I = \{i \in [n-1] \mid \alpha_i > 0\}$. If $|I|<2$, we are trivially done so we can assume w.l.o.g. that $|I| \ge 2$. Select arbitrary $i,j \in I$, $i \neq j$. Let  $\alpha_{ij} = \alpha_i + \alpha_j$ and $\gamma = \alpha_i / \alpha_{ij}$. Observe that $\alpha_i = \gamma \alpha_{ij}$ and $\alpha_j = (1-\gamma) \alpha_{ij}$. Also, as $\alpha_i, \alpha_j > 0$, we have $0 < \gamma < 1$.

Now, let us fix $\alpha_{ij}$ and $\alpha_k$ ($k \notin \{i,j\}$) and focus on $\gamma$. As $\balpha$ is optimal and $\gamma$ is strictly between $0$ and $1$, we have $\frac{\partial LT}{\partial \gamma} = 0$:
\begin{align}\label{prf:rankLinearUp:eq:1}
    &\frac{\partial  LT}{\partial \gamma} = \frac{\partial \int_0^1 \min(B_H,\sum_{k \in [n-1]} \alpha_k \beta_{k}(v)) f(v) dv}{\partial \gamma} = 0 \nonumber\\
    \implies &\frac{\partial \int_0^1 \min(B_H,\sum_{k \in [n-1] \setminus \{i,j\}} \alpha_k \beta_k(v) + \alpha_{ij} (\gamma \beta_i(v) + (1-\gamma) \beta_j(v))) f(v) dv}{\partial \gamma} = 0 \nonumber\\
    \implies &\int_0^{V_H} (\beta_i(v) - \beta_j(v)) f(v) dv = 0 \nonumber\\
    \implies &\int_0^{V_H} \beta_i(v) f(v) dv = \int_0^{V_H} \beta_j(v) f(v) dv.
\end{align}
We use the Leibniz integral rule for the partial derivative w.r.t.\ $\gamma$ above. The non-differentiability of the $\min$ function is not a concern because it happens at only one point of zero measure, and we have an integration over this function; see the footnote for an alternate explanation by a detailed application of the Leibniz integral rule.\footnote{\label{footnote:1}
    We have 
    \[
        \int_0^1 \min(B_H,\beta(v)) f(v) dv = \int_0^{V_H} \beta(v) f(v) dv + \int_{V_H}^1 B_H f(v) dv.
    \]   
    Taking derivative w.r.t.\ $\gamma$, using Leibniz integral rule, we have
    \begin{align*}
        \frac{\partial \int_0^1 \min(B_H,\beta(v)) f(v) dv}{\partial \gamma} &= \frac{\partial \int_0^{V_H} \beta(v) f(v) dv}{\partial \gamma} + \frac{\partial \int_{V_H}^1 B_H f(v) dv}{\partial \gamma}\\
        &= \beta(V_H) f(V_H) \frac{\partial V_H}{\partial \gamma} + \int_0^{V_H} \frac{\partial \beta(v)}{\partial \gamma} f(v) dv - B_H  f(V_H) \frac{\partial V_H}{\partial \gamma}.
    \end{align*}
    As $\beta(V_H) = B_H$, the $\beta(V_H) f(V_H) \frac{\partial V_H}{\partial \gamma}$ and $-B_H  f(V_H) \frac{\partial V_H}{\partial \gamma}$ terms cancel out, and we have 
    \[
        \frac{\partial \int_0^1 \min(B_H,\beta(v)) f(v) dv}{\partial \gamma} = \int_0^{V_H} \frac{\partial \beta(v)}{\partial \gamma} f(v) dv.
    \]
    The argument would be similar if we had both upper and lower thresholds (or equivalently, if we had both $\min$ and $\max$ operations inside the integral).
}

For any $i,j \in I$, we have $\int_0^{V_H} \beta_i(v) f(v) dv = \int_0^{V_H} \beta_j(v) f(v) dv = \int_0^{V_H} \beta(v) f(v) dv$.

Now, we will construct an $\balpha'$ with at most two of its components strictly greater than $0$. The intuition is that we need to satisfy two equations (corresponding to the designer's welfare by $\beta$ below and above $B_H$), which can be achieved by giving positive weights to at most two components of $\balpha$, as shown below. 
As $\beta$ is a convex combination of the $\beta_i$ output functions, we must have one of the following cases:
\begin{enumerate}
    \item $\beta_i(V_H) = \beta(V_H)$ for some $i \in I$. Then set $\alpha'_i = 1$ and  $\alpha'_k = 0$ for $k \in [n-1] \setminus \{i\}$.
    \item $\beta_i(V_H) > \beta(V_H)$ and $\beta_j(V_H) < \beta(V_H)$ for some $i, j \in I$. Because, if $\beta_i(V_H) > \beta(V_H)$ for all $i$, then the convex combination of these $\beta_i(V_H)$ values cannot be equal to $\beta(V_H)$. Set $\alpha'_i = \frac{B_H - \beta_j(V_H)}{\beta_i(V_H) - \beta_j(V_H)}$ and $\alpha'_j = (1 - \alpha'_i)$ and $\alpha'_k = 0$ for $k \in [n-1] \setminus \{i,j\}$.
\end{enumerate}
We can check that $\sum_{k \in [n-1]} \alpha_k \beta_k(V_H) = \sum_{k \in [n-1]} \alpha'_k \beta_k(V_H) = B_H$. So, we get $\beta'^{-1}(B_H) = \beta^{-1}(B_H) = V_H$, where $\beta'$ is the output function induced by $\balpha'$. Also, it is easy to verify that the objective value does not change:
\begin{align*}
    &\int_0^1 \min(B_H,\sum_{j \in [n-1]} \alpha_j \beta_{j}(v)) f(v) dv = \int_0^1 \min(B_H,\sum_{j \in I} \alpha_j \beta_{j}(v)) f(v) dv \\
    &= \int_0^{V_H} \sum_{j \in I} \alpha_j \beta_{j}(v) f(v) dv + B_H(1 - F(V_H)) \\
    &= \sum_{j \in I} \alpha_j  \int_0^{V_H} \beta_{j}(v) f(v) dv + B_H(1 - F(V_H)) \\
    &= \sum_{j \in I} \alpha'_j  \int_0^{V_H} \beta_{j}(v) f(v) dv + B_H(1 - F(V_H)) \quad \text{ (using equation \eqref{prf:rankLinearUp:eq:1})} \\
    &= \int_0^1 \min(B_H,\sum_{j \in [n-1]} \alpha'_j \beta_{j}(v)) f(v) dv.
\end{align*}

\subsubsection{Proof of Theorem~\ref{thm:rankLinear}}
This proof uses the same underlying argument as the proof of Theorem~\ref{thm:rankLinearUp}; the added details are to accommodate the new lower threshold.

Let us assume that $\balpha$ is optimal and $\beta$ is the induced output function. Let $V_{L} = \beta^{-1}(B_L)$ and $V_{H} = \beta^{-1}(B_H)$. Let $V_H = \beta^{-1}(B_H)$. Let $I = \{i \in [n-1] \mid \alpha_i > 0\}$; if $|I|<2$, we are trivially done, so assume w.l.o.g. that $|I| \ge 2$. Select arbitrary $i,j \in I$, $i \neq j$. Let $\alpha_{ij} = \alpha_i + \alpha_j$ and $\gamma = \alpha_i / \alpha_{ij}$. Observe that $\alpha_i = \gamma \alpha_{ij}$ and $\alpha_j = (1-\gamma) \alpha_{ij}$. Also, as $\alpha_i, \alpha_j > 0$, we have $0 < \gamma < 1$.

Let us fix $\alpha_{ij}$ and $\alpha_l$ ($l \notin \{i,j\}$) and focus on $\gamma$. As $\balpha$ is optimal and $\gamma$ is strictly between $0$ and $1$, we have $\frac{\partial LT}{\partial \gamma} = 0$:
\begin{align*}
    &\frac{\partial  LT}{\partial \gamma} = \frac{\partial \int_0^1 \max(B_L, \min(B_H,\sum_{l \in [n-1]} \alpha_l \beta_l(v))) f(v) dv}{\partial \gamma} = 0 \\
    &\implies \frac{\splitfrac{\partial \int_0^1 \max(B_L, \min(B_H, \sum_{k \in [n-1] \setminus \{i,j\}} \alpha_k \beta_k(v)}{+ \alpha_{ij} (\gamma \beta_i(v) + (1-\gamma) \beta_j(v))  )) f(v) dv}}{\partial \gamma} = 0 \\
    &\implies \int_{V_L}^{V_H} (\beta_i(v) - \beta_j(v)) f(v) dv = 0 \\
    &\implies \int_{V_L}^{V_H} \beta_i(v) f(v) dv = \int_{V_L}^{V_H} \beta_j(v) f(v) dv.
\end{align*}
We use the Leibniz integral rule for the partial derivative w.r.t.\ $\gamma$ above. The non-differentiability of the $\min$ and $\max$ functions is not a concern because that happens at only two points of zero measure, and we integrate over this; see footnote~\ref{footnote:1} for an alternate explanation by a detailed application of the Leibniz integral rule.

Thus, for every $i \in I$, we have $\int_{V_L}^{V_H} \beta_i(v) f(v) dv = \int_{V_L}^{V_H} \beta(v) f(v) dv$.

Now, let us look at the constraints that $\balpha$ satisfies:
\begin{enumerate}
    \item \label{thm:rankLinear:const:1} $\beta(V_L) = \sum_{i \in I} \alpha_i \beta_i(V_L) = B_L$;
    \item \label{thm:rankLinear:const:2} $\beta(V_H) = \sum_{i \in I} \alpha_i \beta_i(V_H) = B_H$;
    \item \label{thm:rankLinear:const:3} $\sum_{i \in I} \alpha_i = 1$;
    \item \label{thm:rankLinear:const:4} $\alpha_i \ge 0$ for $i \in I$; 
    \item \label{thm:rankLinear:const:5} $\alpha_i = 0$ for $i \notin I$. 
\end{enumerate}
Observe that any other $\balpha'$ that satisfies all the constraints given above will also be optimal (because $\int_{V_L}^{V_H} \beta_i(v) f(v) dv = \int_{V_L}^{V_H} \beta(v) f(v) dv$ for all $i \in I$, and any $\balpha'$ that satisfies the constraints will have the same value for the objective).
Let us now construct such an $\balpha'$ that satisfies all the constraints above but has at most three strictly positive components. Set $\alpha_i' = 0$ for $i \notin I$ as asked in constraint \ref{thm:rankLinear:const:5} above. Now, constraints \ref{thm:rankLinear:const:1} to \ref{thm:rankLinear:const:4} total $|I|+3$ constraints in $|I|$ variables. This region is a bounded $|I|$-dimensional (linear) polytope. As it is bounded, it has corner points. As it is $|I|$-dimensional, every corner point has $|I|$ tight constraints (satisfied with equality). $|I|-3$ of these tight constraints must be of the form $\alpha'_i = 0$. Hence, selecting such a corner point will give us a solution with at most $3$ of the coordinates strictly greater than $0$.

\subsubsection{Proof of Theorem~\ref{thm:rankSimpleVsOptimal}}
Let us take an optimal solution $\balpha$ with the minimum number of non-zero entries, i.e., the minimum number of indices $i$ with $\alpha_i > 0$.
\begin{enumerate}
    \item If there is only one such index, then we have an approximation ratio of $1$ and we are done. 
    \item Now suppose there are three such indices. 
    Let $i, j, k$ be the three indices for which $\alpha_i, \alpha_j, \alpha_k > 0$ in the optimal solution. W.l.o.g. let $\beta_i(V_L) \le \beta_j(V_L) \le \beta_k(V_L)$. We claim that $\beta_i(V_H) \ge \beta_j(V_H) \ge \beta_k(V_H)$. If this were not true, then for a pair, say $i,j$, we would have had $\beta_i(V_L) \le \beta_j(V_L)$ and $\beta_i(V_H) \le \beta_j(V_H)$, with at least one of the two inequalities strict. As the output functions are single-crossing, we would have $\beta_i(v) \le \beta_j(v)$ for all $v \in [V_L, V_H]$, and we could increase $\alpha_j$ by $\alpha_i$ and decrease $\alpha_i$ to $0$ to get a better solution. 
    
    As $\beta_i(V_L) \le \beta_j(V_L) \le \beta_k(V_L)$ and $\beta_i(V_H) \ge \beta_j(V_H) \ge \beta_k(V_H)$, their convex combination, $\beta$, has: $\beta_i(V_L) \le \beta(V_L) \le \beta_k(V_L)$ and $\beta_i(V_H) \ge \beta(V_H) \ge \beta_k(V_H)$.
    \item Finally, if there are only two positive entries, say $\alpha_i$ and $\alpha_k$, then also we can prove a similar condition: $\beta_i(V_L) \le \beta(V_L) \le \beta_k(V_L)$ and $\beta_i(V_H) \ge \beta(V_H) \ge \beta_k(V_H)$.
\end{enumerate}

\begin{figure}
\centering
    \includegraphics[width=0.75\textwidth]{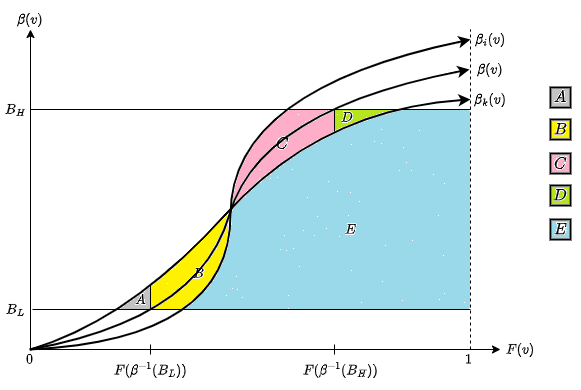}
	\caption{Plot of $\beta_i, \beta_k, \beta$ where $i < k$.}
	\label{fig:one}
\end{figure}

Look at a generic plot of $\beta_i, \beta_k, \beta$ given in Figure~\ref{fig:one}. With reference to the figure, we have:
\begin{itemize}
    \item $\int_0^1 \max(B_L, \min(B_H,\beta_i(v))) f(v) dv = B_L + E + C + D$;
    \item $\int_0^1 \max(B_L, \min(B_H,\beta_k(v))) f(v) dv = B_L + A + B + E$;
    \item $\int_0^1 \max(B_L, \min(B_H,\beta(v))) f(v) dv = B_L F(V_L) + \int_{V_L}^{V_H} \beta(v) f(v) dv + B_H (1-F(V_H))$, which is, by Theorem~\ref{thm:rankLinear}, equal to: $B_L F(V_L) + \int_{V_L}^{V_H} \beta_k(v) f(v) dv + B_H (1-F(V_H)) = B_L + B + E + D$.
\end{itemize}
The approximation ratio is at most
\begin{multline*}
    \frac{B + D + E + B_L}{\max(A + B + E + B_L, C + D + E + B_L)} 
    = \frac{E + B_L + B + D}{E + B_L + \max(A + B, C + D)} \\
    \le \frac{B + D}{\max(A + B, C + D)} 
    \le \frac{B + D}{\max(B, D)} \le 2.
\end{multline*}

\section{Appendix for Section~\ref{sec:opt}}
\subsection{Linear Threshold Objective: Irregular Distributions}\label{sec:app:irregular}

In the study of the optimal linear threshold contest for regular distributions, we used the regularity condition at two places: first, in Lemma~\ref{thm:optLinearRegRight}, to pack the area under $\xi$ to the right; second, in Theorem~$\ref{thm:optLinearReg}$, to prove that we obtain the optimal values for $V_L$ and $V_H$ by solving for the roots of the specific equations given in the theorem statement using an efficient root-finding method. For irregular $F$, first, we give Lemma~\ref{thm:optLinearIrRight} analogous to Lemma~\ref{thm:optLinearRegRight}; second, we find an approximate solution by discretizing the feasible space of $V_L$, $V_H$, and $\xi(V_H)$.\footnote{We would like to note here that the algorithm we provide finds an approximate solution in time polynomial in the reciprocal of the parameter used to discretize $V_L$, $V_H$, and $\xi(V_H)$. One could have discretized the entire optimization problem (the functions are evaluated at discrete values, the integrals are written as finite summations, etc.) and directly found an approximately optimal discretized $\xi$ and $\beta$ using linear programming, also in polynomial time. The advantage of our analysis is that it characterizes the optimal solution and gives a more intuitive algorithm.}

We first introduce some additional notation. Consider the function $\psi_{U,V}(v) = v - \frac{F(V) - F(v)}{f(v)}$, defined on the interval $[U,V]$. We now define $\overline{\psi}_{U,V}$, which is the ironed version of $\psi_{U,V}$, also defined on $[U,V]$. Our definition proceeds in several steps.
\begin{enumerate}
    \item Let $h_{U,V}(y) = \psi_{U,V}(F^{-1}(y))$
    and $H_{U,V}(y) = \int_{F^{-1}(U)}^y h_{U,V}(y) dy$;
    \item Let $\overline{H}_{U,V}(y)$ be the point-wise maximum convex function less than or equal to $H_{U,V}(y)$. Note that at the boundary $\overline{H}_{U,V}(F^{-1}(U)) = H_{U,V}(F^{-1}(U))$ and $\overline{H}_{U,V}(F^{-1}(V)) = H_{U,V}(F^{-1}(V))$;
    \item Let $\overline{h}_{U,V}(y) = \overline{H}_{U,V}'(y)$ 
    and $\overline{\psi}_{U,V}(v) = \overline{h}_{U,V}(F(v))$.
\end{enumerate}
Let $l_{U,V}(v) = \min_{u \in [U,V], \overline{\psi}_{U,V}(u) = \overline{\psi}_{U,V}(v)} u$ and $r_{U,V}(v) = \max_{u \in [U,V], \overline{\psi}_{U,V}(u) = \overline{\psi}_{U,V}(v)} u$, and let $l(u) = l_{0,1}(u)$ and $r(u) = r_{0,1}(u)$.

\begin{lemma}\label{thm:optLinearIrRight}
For an irregular distribution $F$, in the unit-sum setting there is an optimal $\xi$ such that $\int_v^1 \xi(t) f(t) dt = \frac{1}{n}(1-F(v)^n)$ for all $v$ where $\beta(v) < B_H$, $\xi(v) > 0$, and $\overline{H}(F^{-1}(v)) = H(F^{-1}(v))$.
\end{lemma}

\sloppy
Note that Lemma~\ref{thm:optLinearIrRight} does not apply to points $v \in [0,1]$ where $\overline{H}(F^{-1}(v)) < H(F^{-1}(v))$. So, unlike the case for regular distributions, where we had $\int_{V_H}^1 \xi(t) f(t) dt = \frac{1}{n}(1-F(V_H)^n)$ by applying Lemma~\ref{thm:optLinearRegRight} for $v \rightarrow V_H$, we may not have a similar result for irregular distributions.

For regular $F$, the allocation function $\bx$ has three cases depending upon the highest ability, as given in (\ref{sol:linearRegSum}). On the other hand, for irregular $F$, the allocation function $\bx$ has five cases depending upon the highest ability, as given in the theorem below:
\begin{theorem}\label{thm:optLinearIr}
The contest that optimizes the linear threshold objective for irregular distributions and unit-sum constraints has the following allocation function:
\begin{equation}\label{sol:linearIrSum}
    x_i(\bv) = 
    \begin{cases} 0,  &\text{if $\max_j(v_j) < V_L$ or $i \notin W$}\\
    \frac{n \xi(V_L) (F(r_{V_L,V_H}(V_L)) - F(V_L))}{|W_L| (F(r_{V_L,V_H}(V_L))^n - F(V_L)^n)}, &\text{if $ V_L \le \max_j(v_j) < r_{V_L,V_H}(V_L)$ and $i \in W_L$} \\
    1/|W|, &\text{if $r_{V_L,V_H}(V_L) \le \max_j(v_j) < l_{V_L,V_H}(V_H)$ and $i \in W$} \\
    \frac{n \xi(l_{V_L,V_H}(V_H)) (F(V_H) - F(l_{V_L,V_H}(V_H)))}{|W_H|(F(V_H)^n - F(l_{V_L,V_H}(V_H))^n)}, &\text{if $l_{V_L,V_H}(V_H) \le \max_j(v_j) < V_H$ and $i \in W_H$} \\
    \frac{n \xi(V_H) (1 - F(V_H))}{|\widehat{W}|(1 - F(V_H)^n)}, &\text{if $V_H \le \max_j(v_j)$ and $i \in \widehat{W}$}
    \end{cases}
\end{equation}
where $W = \{ k \mid  r_{V_L,V_H}(V_L) \le v_k < l_{V_L,V_H}(V_H), \psi_{V_L,V_H}(v_k) = \max_j(\psi_{V_L,V_H}(v_j)) \}$, $W_L = \{ k \mid v_k \in [V_L,r_{V_L,V_H}(V_L)]\}$, $W_H = \{ k \mid v_k \in [l_{V_L,V_H}(V_H),V_H]\}$, and $\widehat{W} = \{ k \mid v_k \ge V_H \}$. 
\end{theorem}

We can transform the allocation function $\bx(\bv)$ given in the Theorem~\ref{thm:optLinearIr} to an allocation function based on outputs in a manner analogous to (\ref{sol:linearRegSumOutput}). The optimal contest for unit-range can be derived by combining ideas for unit-range with regular $F$ and unit-sum with irregular $F$.

Now, we sketch an approximation algorithm to find the parameters used in Theorem~\ref{thm:optLinearIr}.
We perform an approximate search on the three parameters $V_H$, $\xi(V_H)$, and $V_L$, to maximize the linear threshold objective, using the following algorithm: 
\begin{enumerate}
    \item Select values for $V_H$ and $\xi(V_H)$. 
    \item Assume that $\xi$ is constant in the interval $[l(V_H),V_H)$. Applying Lemma~\ref{thm:optLinearIrRight} to the point $l(V_H)$, compute $\xi(l(V_H))$:
    \begin{multline*}
        \int_{l(V_H)}^1 \xi(t) f(t) dt = \xi(l(V_H))(F(V_H) - F(l(V_H))) + \xi(V_H)(1- F(V_H)) 
        = \frac{1-F(l(V_H))^n}{n(1-F(l(V_H)))} \\
        \implies \xi(l(V_H)) = \frac{1-F(l(V_H))^n}{n(1-F(l(V_H)))(F(V_H) - F(l(V_H)))} - \frac{\xi(V_H)(1- F(V_H))}{F(V_H) - F(l(V_H))}.
    \end{multline*}
    
    \item Compute $\overline{V_L}$ by solving $\beta(V_H) = V_H \xi(V_H) - \int_{\overline{V_L}}^{V_H} \xi(v)dv$ by the bisection method, where the value of $\xi(v)$ for $\overline{V_L} \le v < l(V_H)$ is given as
    \begin{equation*}
        \xi(v) = \begin{cases}
            F(v)^{n-1}, &\text{ if $\overline{H}(F^{-1}(v)) = H(F^{-1}(v))$ or $v \in [l(\overline{V_L}),r(\overline{V_L})]$} \\
            \frac{F(r(v))^n - F(l(v))^n}{n(F(r(v)) - F(l(v)))}, &\text{ if $\overline{H}(F^{-1}(v)) < H(F^{-1}(v))$ and $v > r(\overline{V_L})$}
        \end{cases}
    \end{equation*}
    \item If $\overline{H}(F^{-1}(\overline{V_L})) = H(F^{-1}(\overline{V_L}))$, then set $V_L = \overline{V_L}$, otherwise search for the $V_L \in [l(\overline{V_L}),r(\overline{V_L})]$ (and automatically for $\xi(V_L) \ge B_L$) that satisfies $\int_V^{r(\overline{V_L})} F(v)^{n-1} dv = \int_{V_L}^{r_{V_L,V_H}(V_L)} \xi(V_L) dv + \int_{r_{V_L,V_H}(V_L)}^{r(\overline{V_L})} F(v)^{n-1} dv$ and maximizes $B_L F(V_L) + \int_{V_L}^{r(\overline{V_L})} \beta(v) f(v) dv$. This step selects $V_L$ to optimally redistribute the area under $\xi$ in the interval $[\overline{V_L},r(\overline{V_L})]$.
    \footnote{We skipped a corner case: if $l(V_H)\xi(l(V_H)) < B_L$, then $V_L$ must be greater than $l(V_H)$. Find the $V_L$ in $[l(V_H),V_H]$ that optimizes the objective following a procedure similar to step (4). Also, we need to redistribute $\xi$ in $[l(V_H),V_H)$ according to $\overline{\psi}_{V_L,V_H}$ to find $\xi(l_{V_L,V_H}(V_H))$ in a manner similar to step (2).} 
\end{enumerate}

\subsection{Omitted Proofs}
\subsubsection{Proof of Theorem~\ref{thm:optBinary}}
For unit-range allocations, we optimize $\xi$ subject to the constraints: $0 \le \xi(v) \le 1$. We have 
\[
    \beta(v) = v \xi(v) - \int_0^v \xi(t)dt \le v \xi(v) \le v,
\] 
where the first inequality holds because $\xi(t) \ge 0$ and the second inequality holds because $\xi(t) \le 1$ for all $t \in [0,1]$. We have $\beta(v) \le v \implies B \le \beta^{-1}(B)$. Set $\xi(v) = 0$ for $v < B$ and $\xi(v) = 1$ for $v \ge B$. We have $\xi(B) = 1$ and $\int_0^B \xi(t)dt = 0$, so we get $\beta(B) = B \xi(B) - \int_0^B \xi(t)dt = B \implies \beta^{-1}(B) \le B$. As we have already seen that $\beta^{-1}(B) \ge B$, this is optimal.

\sloppy
For unit-sum, we have an additional constraint on $\xi$, inequality~(\ref{eq:optAvgAlloc1}): $\int_V^1 \xi(v) f(v) dv \le \frac{1-F(V)^n}{n}$ for every $V$. 
\begin{lemma}\label{thm:optBinarySum1}
If $\xi$ is optimal, we can assume w.l.o.g. that $\xi(v) = 0$ for $v < \beta^{-1}(B)$.
\end{lemma}
\begin{proof}
We know that $\beta(V) = V \xi(V) - \int_0^V \xi(x) dx \le V \xi(V)$. Hence, by setting $\xi(v) = 0$ for $v < V$ we still have $\beta^{-1}(B) = V$.
\end{proof}

\begin{lemma}\label{thm:optBinarySum2}
If $\xi$ is optimal, we can assume w.l.o.g. that $\xi(v) = \xi(V)$ for $v \ge \beta^{-1}(B)$, i.e., $\xi(v)$ is constant for $v \ge \beta^{-1}(B)$.
\end{lemma}
\begin{proof}
We define a transformed expected allocation function $\bar{\xi}$, where $\bar{\xi}(v) = \xi(v)$ for $v < V$, and $\bar{\xi}(v) = \frac{1}{1-F(V)} \int_V^1 \xi(t) f(t) dt$ for $v \ge V$. Let $\bar{\beta}$ be the output function for $\bar{\xi}$. As $\xi$ is monotone, $\bar{\xi}(V) \ge \xi(V)$, and therefore $\bar{\beta}(V) \ge \beta(V) \ge B$. So, we still have $\bar{\beta}^{-1}(B) = V$.
\end{proof}

From the previous two lemmas, we know that the allocation function equally distributes the prize among the players who have an ability above some value $V$, where $\beta(v) \ge B$ for $v \ge V$ and $\beta(v) < B$ otherwise. As $\xi$ satisfies $\int_V^1 \xi(v) f(v) dv \le \frac{1-F(V)^n}{n}$, we have:
\[ \xi(V) \le \frac{1 - F(V)^n}{n(1 - F(V))} \implies B \le \beta(V) \le \frac{V(1 - F(V)^n)}{n(1 - F(V))}.\]
As the expression $\frac{V(1 - F(V)^n)}{n(1 - F(V))}$ increases with $V$, and we want to minimize $V$, it is optimal to satisfy the above inequality with an equality. This gives us $B = \frac{V(1 - F(V)^n)}{n(1 - F(V))}$.

Also, as $\frac{V(1 - F(V)^n)}{n(1 - F(V))}$ is continuous and non-decreasing in $[0,1]$, we can efficiently find $V$. However, we do not need to explicitly compute $V$ because the contest that equally distributes the prize to all players who generate an output above $B$, if there are any such players in a given ability profile, automatically induces the required contest.

\subsubsection{Proof of Lemma~\ref{thm:optLinearRegRight}}
We will show that pushing the area under $\xi$ to the right does not decrease the objective. For the initial portion of the proof, let us disregard these two constraints: $\xi$ is monotone and $\int_v^1 \xi(t) f(t) dt \le \frac{1}{n}(1-F(v)^n)$. At the end, we will briefly explain how to incorporate these constraints into the proof.

Take two points $u$ and $v$ such that $u < v$, $\xi(u) > 0$ and $\beta(v) < B_H$. For very small $\delta$ and $\epsilon$ greater than $0$, let us decrease the area under $\xi$ in a small neighborhood $[u-\epsilon,u+\epsilon)$ of $u$ by $\delta/f(u)$ and increase the area under $\xi$ in a small neighborhood $[v-\epsilon,v+\epsilon)$ of $v$ by $\delta/f(v)$. Let $\bar{\xi}$ and $\bar{\beta}$ be the transformed $\xi$ and $\beta$ after the update. Select $\epsilon$ and $\delta$ small enough to maintain $\bar{\xi}(u-\epsilon) \ge 0$ and $\bar{\beta}(v+\epsilon) \le B_H$.

Let us compute the change in the linear threshold objective value. To do that, first, let us look at $\bar{\beta}(y) - \beta(y)$ for $y \in [0,1]$:
\begin{equation*}
    \bar{\beta}(y) - \beta(y) = 
    \begin{cases}
    0, &\text{ if } y \in [0, u-\epsilon) \\
    \frac{-u \delta}{2 \epsilon f(u)}, &\text{ if } y \in [u-\epsilon,u+\epsilon) \\
    \frac{\delta}{f(u)},  &\text{ if } y \in [u+\epsilon,v-\epsilon) \\
    \frac{v \delta}{2 \epsilon f(v)} + \frac{\delta}{f(u)}, &\text{ if } y \in [v-\epsilon,v+\epsilon) \\
    \frac{\delta}{f(u)} - \frac{\delta}{f(v)},  &\text{ if } y \in [v+\epsilon,1] \\
    \end{cases}
\end{equation*}

As we are moving area from left to right, i.e., $u < v$, it is easy to check that for the lower threshold $B_L$, $V_L = \bar{\beta}^{-1}(B_L)$ does not decrease, so neither does $B_L F(V_L)$. For the remaining portion of the linear threshold objective (\ref{eq:optLinearObj}), we have the following:
\begin{align*}
    &\int_{y \ge V_L} (\min(B_H,\bar{\beta}(y)) - \min(B_H,\beta(y))) f(y) dy \\
    &= \int_{y = u-\epsilon}^{u+\epsilon} \frac{-u \delta}{2 \epsilon f(u)} f(y) dy 
        + \int_{y = u+\epsilon}^{v-\epsilon} \frac{\delta}{f(u)} f(y) dy 
        + \int_{y = v-\epsilon}^{v+\epsilon} (\frac{v \delta}{2 \epsilon f(v)} + \frac{\delta}{f(u)}) f(y) dy \\
        & \quad + \int_{y = v+\epsilon}^{1} (\min(B_H,\bar{\beta}(y)) - \min(B_H,\beta(y))) f(y) dy,  \quad \text{ because } \bar{\beta}(v+\epsilon) \le B_H\\
    &= -u \delta + \frac{\delta}{f(u)} (F(v) - F(u)) + v \delta \\
        &\quad + \int_{y = v}^{1} (\min(B_H,\bar{\beta}(y)) - \min(B_H,\beta(y))) f(y) dy,  \quad \text{ as } \epsilon \rightarrow 0.
\end{align*}
We will now consider two cases: (1) $f(u) \le f(v)$ and (2) $f(u) > f(v)$. 

Case (1): $f(u) \le f(v) \implies \delta/f(u) \ge \delta/f(v)$. For $y \ge v+\epsilon$, $\bar{\beta}(y) - \beta(y) = \delta/f(u) - \delta/f(v) \ge 0$, so $\min(B_H,\bar{\beta}(y)) - \min(B_H,\beta(y)) \ge 0$. The total change in the objective is:
\begin{multline*}
    -u \delta + \frac{\delta}{f(u)} (F(v) - F(u)) + v \delta + \int_{y = v}^{1} (\min(B_H,\bar{\beta}(y)) - \min(B_H,\beta(y))) f(y) dy \\
    \ge 0, \quad \text{ as } v > u, F(v) \ge F(u), \text{ and } \min(B_H,\bar{\beta}(y)) \ge \min(B_H,\beta(y)).
\end{multline*}

Case (2): $f(u) > f(v) \implies \delta/f(u) \le \delta/f(v)$. For $y \ge v+\epsilon$, $\bar{\beta}(y) - \beta(y) = \delta/f(u) - \delta/f(v) \le 0$, so $\min(B_H,\bar{\beta}(y)) - \min(B_H,\beta(y)) \ge \frac{\delta}{f(u)} - \frac{\delta}{f(v)}$. The total change in the objective is:
\begin{align*}
    & \quad -u \delta + \frac{\delta}{f(u)} (F(v) - F(u)) + v \delta + \int_{y = v}^{1} (\min(B_H,\bar{\beta}(y)) - \min(B_H,\beta(y))) f(y) dy \\
    &= -u \delta + \frac{\delta}{f(u)} (F(v) - F(u)) + v \delta + \int_{y = v}^{1} (\frac{\delta}{f(u)} - \frac{\delta}{f(v)}) f(y) dy \\
    &= -u \delta + \frac{\delta}{f(u)} (F(v) - F(u)) + v \delta + (\frac{\delta}{f(u)} - \frac{\delta}{f(v)}) (1 - F(v)) \\
    &= \delta(-u + \frac{F(v) - F(u)}{f(u)} + v + \frac{1-F(v)}{f(u)} - \frac{1-F(v)}{f(v)}) \\
    &= \delta(v - \frac{1-F(v)}{f(v)} - (u - \frac{1-F(u)}{f(u)})) \\
    &= \delta(\psi(v) - \psi(u)) \ge 0,  \quad \text{ because $F$ is regular.}
\end{align*}
We proved that transforming $\xi$ in this manner does not decrease the objective value. 

For the rest of the proof, let $a(y) = \int_y^1 \xi(t) f(t) dt$ and $b(y) = (1-F(y)^n)/n$. 

Now, we explain how to incorporate the two constraints we disregarded in the beginning: non-decreasing property of $\xi$ and $a(y) \le b(y)$. For the $a(y) \le b(y)$ constraint: select $v$ to be close to the rightmost\footnote{If the set is open, we select $v$ in limit. Same for $u$.} point for which $a(v) < b(v)$ and has $\beta(v) < B_H$, select $u$ to be close to the rightmost point less than $v$ for which $a(u) = b(u)$ and $\beta(u) > 0$, if there is no such point $u$ less than $v$ with $a(u) = b(u)$, then select any point with $\beta(u) > 0$. If we select a small enough area to move from $u$ to $v$ (parameterized by $\delta$ in the first part of the proof) and select the neighborhood of $u$ and $v$ suitably (parameterized by $\epsilon$ earlier) we can satisfy the constraint. For the monotonically non-decreasing property, selecting $v$ to be close to the rightmost point with $a(v) < b(v)$ and $\beta(v) < B_H$, and increasing $\xi$ in very small increments, and repeating until convergence, maintains the non-decreasing property of $\xi$ in the aggregate.

Starting from an arbitrary $\xi$, one can reach a $\xi$ that satisfies the condition given in the statement of the lemma by transformations to $\xi$ as given above. This completes the proof.

\subsubsection{Proof of Lemma~\ref{thm:optLinearRegUp}}
Let us assume that the lemma is false, which gives us $\int_{V_H}^1 \beta(v) f(v) dv > B_H (1 - F(V_H))$.

We average out $\xi$ in the interval $v \in [V_H-\delta, 1]$, for some $\delta > 0$, while maintaining $\beta(v) \ge B_H$ for $v \ge V_H - \delta$. This will improve our objective. We can check that the transformed $\xi$ satisfies $\xi(v) = \frac{1}{1-F(V_H-\delta)} \int_{V_H-\delta}^1 \xi(v) f(v) dv$ for $v \ge V_H-\delta$. We can find the $\delta$ by solving the following equation:
\[
B_H = \beta(V_H-\delta) = (V_H-\delta) \frac{1}{1-F(V_H-\delta)} \int_{V-\delta}^1 \xi(v) f(v) dv - \int_0^{V_H-\delta}\xi(v) dv.
\]
Observe that the right-hand side is a non-increasing continuous function of $\delta$. As $\delta$ goes from $0$ to $V_H$, the right-hand side goes from strictly above $B_H$ to $0$, so we obtain the required solution for $\delta$.

\subsubsection{Proof of Theorem~\ref{thm:optLinearReg}}
We provide separate proofs for unit-range and unit-sum settings.

\paragraph{Unit-Range} Note that Lemmas~\ref{thm:optLinear1},\ref{thm:optLinear2}, and~\ref{thm:optLinearRegUp} are  applicable for unit-range, and Lemma~\ref{thm:optLinearRegRight} is applicable with slight modification. Together, they imply that $\xi(v) = 0$ for $v < V$ and $\xi(v) = 1$ for $v \ge V$, for some $V \in [0,1]$. From Lemma~\ref{thm:optLinear1} we also have $V_L = V \ge B_H$, and from Lemma~\ref{thm:optLinearRegUp}, $V \le V_H \le B_H$ if there exists a $V_H$, or $V < B_H$ if not. The objective value can be written as:
\[
    B_L F(V) + V (1-F(V)).
\]
Differentiating w.r.t. $V$ and equating to $0$, we obtain
\begin{multline*}
    B_L f(V) - Vf(V) + (1-F(V)) = f(V) (B_L - (V - \frac{1-F(V)}{f(V)}))
    = f(V) (B_L - \psi(V)) = 0 \\
    \implies \psi(V) = B_L \implies V = \psi^{-1}(B_L).
\end{multline*}
We can also observe that the solution to the above equation is a global maximum because the derivative of the objective is greater than $0$ for $V < \psi^{-1}(B_L)$ and less than $0$ afterwards. Plugging in the constraints on $V$, $B_L \le V \le B_H$, the optimal solution is $V = \max(B_L,\min(B_H,\psi^{-1}(B_L)))$.

\paragraph{Unit-Sum} We divide the analysis into two cases, depending on whether $\beta$ hits the upper threshold $B_H$.
\begin{enumerate}
    \item $\beta(v) < B_H$ for $v \in [0,1]$. We do not have a $V_H$, and for $V_L$ we have the following inequality:
\begin{multline*}
     \beta(1) < B_H \implies 1 \xi(1) - \int_{V_L}^1 \xi(v) dv < B_H \implies \int_{V_L}^1 F(v)^{n-1} dv > 1-B_H 
     \implies V_L < V_{\mmid},
\end{multline*}
    where $V_{\mmid}$ is the solution of the equation $\int_{V_{\mmid}}^1 F(v)^{n-1} dv = 1 - B_H$. We also know that
    \[
        \beta(V_L) \ge B_L \implies V_L \xi(V_L) = V_L F(V_L)^{n-1} \ge B_L \implies V_L \ge V_{\low},
    \]
    where $V_{\low}$ is the solution to $V_{\low} F(V_{\low})^{n-1} = B_L$. Thus, a value of $V_L$ in $[V_{\low},V_{\mmid})$ maximizes the objective:
    \[
        \mathit{OBJ} = B_L F(V_L) + \int_{V_L}^1 \beta(v) f(v) dv.
    \]
    Now, $\beta(v) = v F(v)^{n-1} - \int_{V_L}^v F(t)^{n-1} dt \implies \frac{d\beta(v)}{dV_L} = F(V_L)^{n-1}$. Differentiating $\mathit{OBJ}$ w.r.t. $V_H$ we get:
    \begin{align*}
        \frac{d \mathit{OBJ}}{d V_H} &= B_L f(V_L) - \beta(V_L) f(V_L) + \int_{V_L}^1 \frac{d\beta(v)}{dV_L} f(v) dt \\
        &= B_L f(V_L) - V_L F(V_L)^{n-1} f(V_L) + F(V_L)^{n-1} (1 - F(V_L)) \\
        &= f(V_L) F(V_L)^{n-1} (\frac{B_L}{F(V_L)^{n-1}}  - V_L +  \frac{1 - F(V_L)}{f(V_L)}) \\
        &= f(V_L) F(V_L)^{n-1} (\frac{B_L}{F(V_L)^{n-1}}  - \psi(V_L)).
    \end{align*}
    As $\frac{B_L}{F(V_L)^{n-1}}$ decreases with $V_L$, $\psi(V_L)$ increases with $V_L$, and $f(V_L) F(V_L)^{n-1}$ is non-negative, the root of $\frac{B_L}{F(V_L)^{n-1}}  - \psi(V_L) = 0$
    is the global maximum. Also, as the function is continuous, we can efficiently find a solution using a root finding algorithm such as the bisection method; let the solution be $\overline{V_L}$. The optimal $V_L$ for this case will be $V_L = \max(V_{\low}, \min(V_{\mmid}, \overline{V_L}))$.
    
    \item $\beta(v) \ge B_H$ for $v \ge V_H \in [0,1]$. We have the following equality:
    \[
        \beta(V_H) =  V_H \xi(V_H) - \int_{V_L}^{V_H} \xi(v) dv = V_H \eta(F(V_H)) - \int_{V_L}^{V_H} F(v)^{n-1} dv = B_H,
    \]
    where $\eta(x) = \frac{1-x^n}{n(1-x)}$. Note that $\eta'(x) = \frac{1-x^n}{n(1-x)^2} - \frac{nx^{n-1}}{1-x} = \frac{1}{1-x}(\eta(x) - x^{n-1})$ and also that $\eta(x) \ge x^{n-1}$ and $\eta'(x) \ge 0$ for $x \in [0,1]$. 
    
    For $u \ge v$, $\psi_u(v) = v - \frac{F(u) - F(v)}{f(v)}$. Observe that $\psi_u(v)$ is non-decreasing in $v$ because $\frac{\psi_u(v)}{dv} = 2 + \frac{(F(u) - F(v)) f'(v)}{f(v)^2}$ is obviously non-negative if $f'(v) \ge 0$, and if $f'(v) < 0$, then $\frac{\psi_u(v)}{dv}  = 2 + \frac{(F(u) - F(v)) f'(v)}{f(v)^2} \ge 2 + \frac{(1 - F(v)) f'(v)}{f(v)^2} = \frac{\psi(v)}{dv} \ge 0$.
    
    As $\beta(V_H) = B_H$, we get $V_H \in [V_{\uup},1]$ and $V_L \in [V_{\mmid}, V_{\uup}]$ where $V_{\uup}$ is the solution of the equation $B_H = V_{\uup} \eta(F(V_{\uup}))$. Differentiating $\beta(V_H) = B_H$ w.r.t. $V_L$ we get:
    \begin{align*}
        &\frac{d\beta(V_H)}{d V_L} = \frac{dB_H}{d V_L} = 0\\
        &\implies (V_H \eta'(F(V_H)) f(V_H) + \eta(F(V_H)) - F(V_H)^{n-1}) \frac{d V_H}{d V_L} + F(V_L)^{n-1}  = 0 \\
        &\implies (V_H \eta'(F(V_H)) f(V_H) + (1-F(V_H)) \eta'(F(V_H))) \frac{d V_H}{d V_L} + F(V_L)^{n-1}  = 0 \\
        &\implies \frac{d V_H}{d V_L} = \frac{ - F(V_L)^{n-1}}{ \eta'(F(V_H)) (V_H f(V_H) + (1-F(V_H)))} \le 0.
    \end{align*}
    Given $V_L$ and $V_H$, the objective can be written as:
    \begin{align*}
        \mathit{OBJ} = B_L F(V_L) + \int_{V_L}^{V_H} \beta(v) f(v) dv + B_H (1 - F(V_H)).
    \end{align*}
    Differentiating $\mathit{OBJ}$ w.r.t. $V_L$ we get:
    \begin{align*}
        \frac{d \mathit{OBJ}}{d V_L} &= B_L f(V_L) + \beta(V_H) f(V_H) \frac{d V_H}{d V_L} - \beta(V_L) f(V_L) \\
            &\quad + \int_{V_L}^{V_H} F(V_L)^{n-1}  f(v) dv - B_H f(V_H) \frac{d V_H}{d V_L} \\
        &= B_L f(V_L)  -  V_L F(V_L)^{n-1} f(V_L) + \int_{V_L}^{V_H} F(V_L)^{n-1} f(v) dv \\
        &\quad + (\beta(V_H) - B_H) f(V_H) \frac{d V_H}{d V_L} \\
        &= B_L f(V_L)  -  V_L F(V_L)^{n-1} f(V_L) + \int_{V_L}^{V_H} F(V_L)^{n-1} f(v) dv \\ 
        &= F(V_L)^{n-1} f(V_L) (\frac{B_L}{F(V_L)^{n-1}} - V_L + \frac{F(V_H) - F(V_L)}{f(V_L)}) \\
        &= F(V_L)^{n-1} f(V_L) (\frac{B_L}{F(V_L)^{n-1}} - \psi_{V_H}(V_L)).
    \end{align*}
    From the equation above, to find the solution of $\frac{d \mathit{OBJ}}{d V_L} = 0$, we need to solve for the values of $V_L$ and $V_H$ that satisfy $\frac{B_L}{F(V_L)^{n-1}} - \psi_{V_H}(V_L) = 0$ (and $\beta(V_H) = B_H$). As the $F$ and $\psi_{V_H}$ are continuous, we can efficiently find a solution using a root finding algorithm. Moreover, the pair of values for $V_L$ and $V_H$ that satisfies $\frac{B_L}{F(V_L)^{n-1}} - \psi_{V_H}(V_L) = 0$ is optimal because:
    \begin{itemize}
        \item The first term, $\frac{B_L}{F(V_L)^{n-1}}$, decreases with $V_L$.
        \item The second term, $\psi_{V_H}(V_L)$, has a derivative: $\frac{d\psi_{V_H}(V_L)}{dV_L} = \frac{\partial \psi_{V_H}(V_L)}{\partial V_H} \frac{d V_H}{d V_L} + \frac{\partial \psi_{V_H}(V_L)}{\partial V_L}$. As $\frac{\partial \psi_{V_H}(V_L)}{\partial V_H} = \frac{-f(V_H)}{f(V_L)} \le 0$ and $\frac{d V_H}{d V_L} \le 0$, we get $\frac{\partial \psi_{V_H}(V_L)}{\partial V_H} \frac{d V_H}{d V_L} \ge 0$. Also,  $\frac{\partial \psi_{V_H}(V_L)}{\partial V_L} \ge 0$ as shown earlier. So, $\psi_{V_H}(V_L)$ is a non-decreasing function of $V_L$.
    \end{itemize}
    Let the values of $V_L$ and $V_H$ that satisfy $\frac{B_L}{F(V_L)^{n-1}} - \psi_{V_H}(V_L) = 0$ be $\overline{V_L}$ and $\overline{V_H}$, respectively. Overall, we have the optimal $V_L = \min(V_{\uup},\max(V_{\mmid}, \overline{V_L}))$ and the optimal $V_H = \min(1,\max(V_{\uup}, \overline{V_H}))$.
    
\end{enumerate}
One of the two cases, either $\beta(v)$ touches the upper threshold $B_H$ or it does not, will give us the overall optimal solution. 

Given the optimal expected allocation function $\xi(v)$, we can easily derive the optimal allocation function $\bx(\bv)$, given in the theorem statement.

\subsubsection{Proof Sketch of Lemma~\ref{thm:optLinearIrRight}}
The proof is very similar to Lemma~\ref{thm:optLinearRegRight}. The main modifications are: first, we can check that in $[l(v), r(v)]$ if we flatten $\xi$, i.e., we set $\xi(y) = \frac{\int_{l(v)}^{r(v)} \xi(t) f(t) dt}{F(r(v)) - F(l(v))}$ for $y \in [l(v), r(v)]$, then we do not decrease the objective;\footnote{We do this for all points except for the points in $[l(V_L), r(V_L)]$ if $l(V_L) < V_L$, otherwise it might change $\beta^{-1}(B_L)$. Note that the statement of the lemma accommodates for this.} second, we account for the change in the objective for transferring area under $\xi$ from $[l(u),r(u)]$ to $[l(v),r(v)]$, in aggregate, rather than from $u$ to $v$ as we did in Lemma~\ref{thm:optLinearRegRight}.

\end{document}

%% file: macros.tex
\usepackage{amsmath,amsthm,amsfonts,amssymb,float,bm,bbm}

\usepackage{hyperref}
\usepackage[svgnames]{xcolor}
\hypersetup{colorlinks={true},urlcolor={blue},linkcolor={DarkBlue},citecolor=[named]{DarkGreen}}

\usepackage{microtype}
\usepackage[capitalise,nameinlink,noabbrev]{cleveref}

\usepackage{xspace}
\usepackage{doi}
\usepackage{graphicx}
\usepackage{mathtools}
\usepackage{subcaption}
\usepackage{todonotes}
\usepackage{authblk}
\usepackage{natbib}

\newtheorem{theorem}{Theorem}[section]
\newtheorem{lemma}[theorem]{Lemma}

\theoremstyle{definition}
\newtheorem{definition}{Definition}[section]
\newtheorem{example}{Example}[section]

\graphicspath{ {./images/} }
\usepackage{mathtools}
\usepackage{amsmath,amsfonts,bm,bbm}

\newcommand{\bx}{\bm{x}}
\newcommand{\bv}{\bm{v}}
\newcommand{\bw}{\bm{w}}
\newcommand{\bb}{\bm{b}}
\newcommand{\balpha}{\bm{\alpha}}
\DeclareMathOperator{\b1}{\mathbbm{1}}
\newcommand{\bR}{\mathbb{R}}
\newcommand{\bE}{\mathbb{E}}
\newcommand{\bP}{\mathbb{P}}
\newcommand{\bZ}{\mathbb{Z}}

\DeclareMathOperator*{\argmax}{arg\,max}
\DeclareMathOperator*{\argmin}{arg\,min}

\newcommand{\low}{\mathrm{low}}
\newcommand{\mmid}{\mathrm{mid}}
\newcommand{\uup}{\mathrm{up}}